\documentclass{llncs}
\usepackage[T1]{fontenc}
\usepackage{microtype}
\usepackage{amssymb}
\usepackage{xcolor}
\usepackage{xspace}
\usepackage{graphicx}
\usepackage{xspace}
\usepackage{hyperref}
\usepackage{lineno}
\usepackage{wrapfig}
\usepackage{stackengine}
\usepackage{mathtools}

\graphicspath{{figures/}}

\newcommand{\etal}{et~al.\xspace}
\newcommand{\subcap}[1]{\textbf{(#1)}}
\newcommand{\mypar}[1]{\smallskip\noindent{\bfseries #1.}}

\stackMath
\newcommand\frightarrow{\scalebox{.8}[.8]{$\blacktriangleright$}}
\newcommand\frightarrowsmall{\scalebox{.5}[.5]{$\blacktriangleright$}}
\newcommand\morpheq{\mathrel{%
  \stackengine{-4.975pt}{=}{\frightarrowsmall}{U}{r}{F}{T}{S}}}
 \newcommand\morph{\mathrel{%
  \stackengine{-4.95pt}{\relbar\joinrel\relbar}{\frightarrow}{U}{r}{F}{T}{S}}}

\title{Optimal Morphs\\ of Planar Orthogonal Drawings~II
} 

\author{Arthur van Goethem
\and Bettina Speckmann
\and Kevin Verbeek
}

\institute{Department of Mathematics and Computer Science, TU Eindhoven, The Netherlands
\\ \email{[a.i.v.goethem|b.speckmann|k.a.b.verbeek]@tue.nl}
}

\let\doendproof\endproof
\renewcommand\endproof{~\hfill$\qed$\doendproof}

\begin{document}


\maketitle

\begin{abstract}
Van Goethem and Verbeek~\cite{vanGoethemVerbeek2018} recently showed how to morph between two planar orthogonal drawings $\Gamma_I$ and $\Gamma_O$ of a connected graph $G$ while preserving planarity, orthogonality, and the complexity of the drawing during the morph. Necessarily drawings $\Gamma_I$ and $\Gamma_O$ must be equivalent, that is, there exists a homeomorphism of the plane that transforms $\Gamma_I$ into $\Gamma_O$. Van Goethem and Verbeek use $O(n)$ linear morphs, where $n$ is the maximum complexity of the input drawings. However, if the graph is disconnected their method requires $O(n^{1.5})$ linear morphs. In this paper we present a refined version of their approach that allows us to also morph between two planar orthogonal drawings of a disconnected graph with $O(n)$ linear morphs while preserving planarity, orthogonality, and linear complexity of the intermediate drawings.

Van Goethem and Verbeek measure the structural difference between the two drawings in terms of the so-called \emph{spirality} $s = O(n)$ of $\Gamma_I$ relative to $\Gamma_O$ and describe a morph from $\Gamma_I$ to $\Gamma_O$ using $O(s)$ linear morphs. We prove that $s+1$ linear morphs are always sufficient to morph between two planar orthogonal drawings, even for disconnected graphs. The resulting morphs are quite natural and visually pleasing.

\end{abstract}

\section{Introduction}

Continuous morphs of planar drawings have been studied for many years, starting as early as 1944, when Cairns~\cite{Cairns1944} showed that there exists a planarity-preserving continuous morph between any two (compatible) triangulations that have the same outer triangle. These results were extended by Thomassen~\cite{Thomassen1983} in 1983, who gave a constructive proof of the fact that two compatible straight-line drawings can be morphed into each other while maintaining planarity. The resulting algorithm to compute such a morph takes exponential time (just as Cairns' result). Thomassen also considered the orthogonal setting and showed how to morph between two rectilinear polygons with the same turn sequence. For planar straight-line drawings the question was settled by Alamdari~\etal~\cite{Alamdari2016}, following work by Angelini~\etal~\cite{Angelini2013}. They showed that $O(n)$ uni-directional linear morphs are sufficient to morph between any compatible pair of planar straight-line drawings of a graph with $n$ vertices while preserving planarity. The corresponding morph can be computed in $O(n^3)$ time.

In this paper we consider the orthogonal setting, that is, we study planarity-preserving morphs between two planar orthogonal drawings $\Gamma_I$ and $\Gamma_O$ with maximum complexity $n$, of a graph $G$. Here the complexity of an orthogonal drawing is defined as the number of vertices and bends. All intermediate drawings must remain orthogonal, as to not disrupt the mental map of the reader. This immediately implies that the results of Alamdari~\etal~\cite{Alamdari2016} do not apply, since they do not preserve orthogonality. Biedl~\etal~\cite{Biedl2006} described the first results in this setting, for so-called \emph{parallel} drawings, where every edge has the same orientation in both drawings. They showed how to morph between two parallel drawings using $O(n)$ linear morphs while maintaining parallelity and planarity. More recently, Biedl~\etal~\cite{Biedl2013} showed how to morph between two planar orthogonal drawings using $O(n^2)$ linear morphs, while preserving planarity, orthogonality, and linear complexity. Van Goethem and Verbeek~\cite{vanGoethemVerbeek2018} improved this bound further to $O(n)$ linear morphs for a connected graph $G$. This bound is tight, based on the lower bound for straight-line graphs proven by Alamdari~\etal~\cite{Alamdari2016}.

If the graph $G$ is disconnected, then Aloupis~\etal~\cite{Aloupis2015} show how to connect $G$ in a way that is compatible with both $\Gamma_I$ and $\Gamma_O$ while increasing the complexity of the drawings to at most $O(n^{1.5})$.
They also prove a matching lower bound if $G$ has at most $\frac{n}{4}$ connected components.
This directly implies that Van Goethem and Verbeek require $O(n^{1.5})$ linear morphs for a disconnected graph $G$. 

\mypar{Paper Outline} We show how to refine the approach by Van Goethem and Verbeek~\cite{vanGoethemVerbeek2018} to also morph between two planar orthogonal drawings of a disconnected graph $G$ using $O(n)$ linear morphs while preserving planarity, orthogonality, and linear complexity. In Section~\ref{sec:prelim} we describe the necessary background. In particular, we discuss \emph{wires}: equivalent sets of horizontal and vertical polylines that capture the $x$- and $y$-order of the vertices in $\Gamma_I$ and $\Gamma_O$. The \emph{spirality} of these wires guides the morph. In Section~\ref{sec:disconnected} we show how to find sets of wires with linear spirality for equivalent orthogonal planar drawings $\Gamma_I$ and $\Gamma_O$ of a disconnected planar graph $G$. Van Goethem and Verbeek are agnostic of the connectivity of the graph once they create the wires. Hence, using the wires constructed in Section~\ref{sec:disconnected}, we can directly apply their approach to disconnected graphs.

In the remainder of the paper we show how to ``batch'' intermediate morphs. We argue solely based on sets of wires, hence the results apply to both connected and disconnected graphs. In particular, in Section~\ref{sec:noConstant} we show how to combine all intermediate morphs that act on segments of spirality $s$ into one single linear morph. Hence we need only $s$ linear morphs to morph from $\Gamma_I$ to $\Gamma_O$. However, the rerouting and simplification operations introduced by van Goethem and Verbeek to lower the intermediate complexity are not compatible with batched linear morphs and hence intermediate drawings have complexity of $O(n^3)$. In Section~\ref{sec:linearComplexity} we  present refined versions of both operations which allow us to maintain linear complexity through the $s$ linear morphs. The initial setup for these operations costs one additional morph, for a total of $s+1$ linear morphs that preserve planarity, orthogonality, and linear complexity. We implemented our algorithm and believe that the resulting morphs are natural and visually pleasing\footnote{See \texttt{https://youtu.be/n0ZaPtfg9TM} for a short movie.}. 
We restrict our arguments to proof sketches, full proofs can be found in the appendix.

\section{Preliminaries}\label{sec:prelim}
\mypar{Orthogonal drawings} 
A \emph{drawing} $\Gamma$ of a graph $G=(V,E)$ is a mapping from every vertex $v\in V$ to a unique point $\Gamma(v)$ in the Euclidean plane and from each edge $(u,v)$ to a simple curve in the plane starting at $\Gamma(u)$ and ending at $\Gamma(v)$.
A drawing is \emph{planar} if no two curves intersect in an internal point, and no vertices intersect a curve in an internal point.
A drawing is \emph{orthogonal} if each edge is mapped to an orthogonal polyline consisting of horizontal and vertical segments meeting at \emph{bends}. 
In a \emph{straight-line drawing} every edge is represented by a  single line-segment.
Two planar drawings $\Gamma$ and $\Gamma'$ are \emph{equivalent} if there exists a homeomorphism of the plane that transforms $\Gamma$ into $\Gamma'$.

We consider morphs between two equivalent drawings of a graph $G$. 
To simplify the presentation, we assume that both drawings are straight-line drawings with $n$ vertices.
If this is not the case then we first \emph{unify} $\Gamma$ and $\Gamma'$. We subdivide segments, creating additional virtual bends, to ensure that every edge is represented by the same number of segments in $\Gamma$ and $\Gamma'$. Next, we replace all bends with vertices. All edges of the resulting graph $G^*$ are now represented by straight segments (horizontal or vertical) in both $\Gamma$ and $\Gamma'$. 

\begin{figure}[b]
\centering
\includegraphics[scale=0.85]{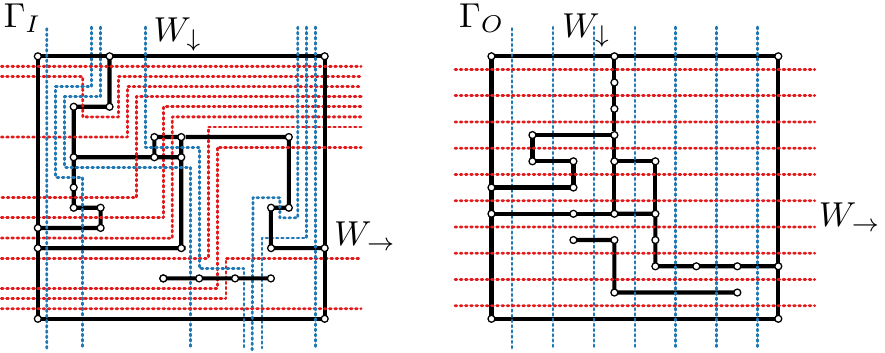}
\caption{Two unified drawings $\Gamma_I$ and $\Gamma_O$ of $G$ (black) plus
equivalent wires (red/blue).}
\label{fig:wires}
\end{figure}

A \emph{linear morph} of two drawings $\Gamma$ and $\Gamma'$ can be described by a continuous linear interpolation of all vertices and bends, which are connected by straight segments.
For each $0\le t\le 1$ there exists an intermediate drawing $\Gamma_t$ where each vertex $v$ is drawn at $\Gamma_t(v)=(1-t)\Gamma_v + t\Gamma'_v$ ($\Gamma_0 = \Gamma$ and $\Gamma_1 = \Gamma'$).
A linear morph \emph{maintains planarity} (orthogonality, linear complexity, resp.), if every intermediate drawing $\Gamma_t$ is planar (orthogonal, of linear complexity, resp.).

\mypar{Wires}
Following van Goethem and Verbeek~\cite{vanGoethemVerbeek2018} we use orthogonal polylines called \emph{wires} as the main tool to determine the morph.
Wires consist of horizontal or vertical segments called \emph{links}.
We use two sets of wires to capture the horizontal and vertical order of the vertices in $\Gamma_I$ and $\Gamma_O$.
The \emph{lr-wires} $W_\rightarrow$ traverse the drawings from left to right, and the \emph{tb-wires} $W_\downarrow$ traverse the drawings from top to bottom.
Since the horizontal and vertical order of the vertices in $\Gamma_O$ are guiding our morph, the wires $W_\rightarrow$ and $W_\downarrow$ are simply horizontal and vertical lines in $\Gamma_O$ separating consecutive vertices in the $x$- and $y$-order (only if their $x$- or $y$-coordinates are distinct).
$\Gamma_O$ and $\Gamma_I$ are equivalent, hence there exist wires in $\Gamma_I$ that are equivalent to the wires in $\Gamma_O$: there is a one-to-one matching between the wires of $\Gamma_O$ and $\Gamma_I$ such that matching wires partition the vertices identically, and cross both the segments of the drawings and the links of the other wires in the same order (see Fig.~\ref{fig:wires}).
Any such two wires in $\Gamma_I$ do not cross if they are from the same set and cross exactly once otherwise. 

Van Goethem and Verbeek use the \emph{spirality} of wires as a measure for the distance to $\Gamma_O$ (where all wires are straight lines of spirality zero).
Spirality is a well-established measure in the context of orthogonal drawings and is frequently used for bend-minimization~\cite{Blasius2016,DiBattista1998,Didimo2009}.
Specifically, let $w\in W_\rightarrow$ be a lr-wire, and $\ell_1,\ldots,\ell_k$ be the links ordered along $w$.
Let $b_i$ be the orientation of the bend from $\ell_i$ to $\ell_{i+1}$, where $b_i=1$ for a left turn, $b_i=-1$ for a right turn, and $b_i=0$ otherwise.
The \emph{spirality} of a link $\ell_i$ is defined as $s(\ell_i)=\sum_{j=1}^{i-1}b_i$.
A \emph{maximum-spirality link} is any link with the largest absolute spirality.
The spirality of a wire is the maximum absolute spirality of any link in the wire, the spirality of a set of wires is the maximum spirality of any wire in the set.

The spirality of a drawing $\Gamma$ is not well defined: it is always relative to another drawing $\Gamma'$ and the straight-line wires induced by $\Gamma'$.
Furthermore, there are possibly multiple sets of matching wires in $\Gamma$ for the straight-line wires in $\Gamma'$.
Still, whenever the drawing $\Gamma'$ and the matching set of wires in $\Gamma$ are clear from the context, then by abuse of notation we will speak of the spirality of $\Gamma$. Unless stated otherwise, we always consider spirality relative to $\Gamma_O$.


\begin{figure}[b]
\centering
\includegraphics[scale=0.85]{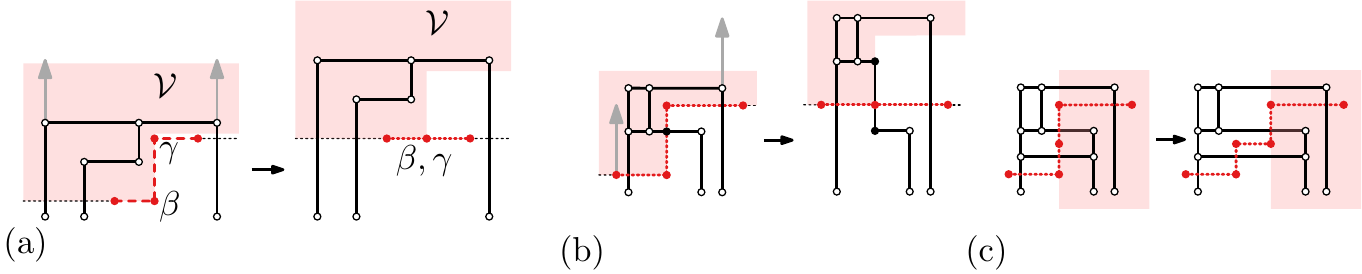}
\caption{A drawing (black) with vertices (open marks) and bends (closed marks). \subcap{a} A \emph{zigzag-eliminating slide} with center link $\overline{\beta\gamma}$. \subcap{b} Introducing two additional bends in a crossing segment ensures orthogonality. \subcap{c} A \emph{bend-introducing slide}.}
\label{fig:equivalentWires}
\end{figure}

\mypar{Slides}
Biedl~\etal~\cite{Biedl2013} introduced \emph{slides} as a particular type of linear morph that operates on the segments of the drawing.
Van Goethem and Verbeek~\cite{vanGoethemVerbeek2018} extended this concept to wires.
Slides on wires may be accompanied by the insertion or deletion of bends in the drawing.
In the following we exclusively consider slides on wires.
A \emph{zigzag} consists of three consecutive links of a wire and two bends $\beta$ and $\gamma$ that form a left turn followed by a right turn or vice versa.
Consider the horizontal zigzag with bends $\beta$ and $\gamma$ in Figure~\ref{fig:equivalentWires}(a).
Let $\mathcal{V}$ be the set of vertices and bends of both the drawing and the wires that are (1) above or at the same height as $\beta$ and strictly to the left of $\beta$, (2) that are strictly above $\gamma$, and (3) $\beta$. The corresponding region is shaded in Figure~\ref{fig:equivalentWires}. A \emph{zigzag-eliminating slide} is a linear morph that straightens a zigzag on a wire by moving all vertices and bends in $\mathcal{V}$ up by the initial distance between $\beta$ and $\gamma$.

By definition, wires do not contain any vertices or bends of the drawing or other wires. However, the center link $\overline{\beta\gamma}$ might be crossed by a segment of the drawing or a link of a wire in the other set (see Fig.~\ref{fig:equivalentWires}(b) for a crossing with a segment of the drawing). In this case we introduce two virtual bends in the segment or the link on the crossing and symbolically offset one to the right and one to the left.
The left bend is thus included in $\mathcal{V}$ while the right bend is not. We can prevent that multiple segments or links cross $\overline{\beta\gamma}$ using so-called \emph{bend-introducing slides} as discussed in~\cite{vanGoethemVerbeek2018} (see Fig.~\ref{fig:equivalentWires}(c)).

\section{Linear morphs for disconnected graphs}\label{sec:disconnected}

Let $\Gamma_I$ and $\Gamma_O$ be two equivalent planar orthogonal drawings of a disconnected graph $G$.
For a connected graph there is a unique homotopy class in $\Gamma_I$ that contains all possible wires that match a given wire $w$ from $\Gamma_O$. This statement does not hold for disconnected graphs: there might be more than one homotopy class in $\Gamma_I$ that matches $w$ (see Fig.~\ref{fig:consensus}(a)). If we choose homotopy classes independently for the wires in $\Gamma_I$ then their union might not be equivalent to the set of wires in $\Gamma_O$, for example, wires might cross more than once (see Fig.~\ref{fig:consensus}(c)).

Below we show that we can choose homotopy classes for the wires in $\Gamma_I$ incrementally, first for the lr-wires and then for the tb-wires, while maintaining the correct intersection pattern and hence equivalence with $\Gamma_O$. For each of the resulting equivalence classes we add the shortest wire to the set of wires. It remains to argue that the resulting set of wires has spirality $O(n)$ despite the interdependence of the homotopy classes and the fact that the arrangement of drawing and wires can have super-linear complexity (which invalidates the proofs from~\cite{vanGoethemVerbeek2018}). Below we consider only $W_\rightarrow$, analogous results hold for $W_\downarrow$.

\begin{figure}[b]
\centering
\includegraphics[scale=0.85]{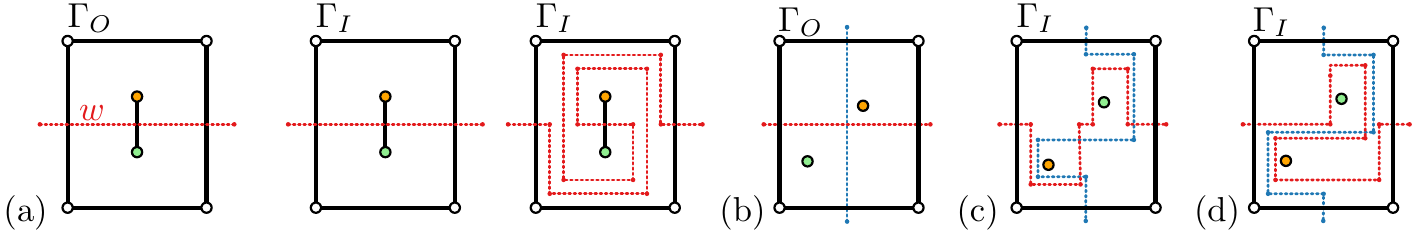}
\caption{\subcap{a} A (straight-line) wire $w$ in $\Gamma_O$ (red) and two possible wires in $\Gamma_I$ from different homotopy classes that both match $w$. \subcap{b} A graph with three connected components. \subcap{c} Wires in $\Gamma_I$ that cross three times. \subcap{d} Set of wires equivalent to $\Gamma_O$.}
\label{fig:consensus}
\end{figure}

\begin{lemma}\label{lem:spiral}
For each right-oriented link $\ell_\rightarrow$ of a wire $w\in W_\rightarrow$ with positive (negative) spirality $s$ there exists a vertical line $L$ and a subsequence of $\Omega(|s|)$ links of $w$ crossing $L$, such that the absolute spiralities of the links in sequence are $[0,2,4,\ldots,|s|-2,|s|]$, and when ordered top-to-bottom (bottom-to-top) along $L$ form the sequence $[2,6,10,\ldots, |s|-2,|s|,|s|-4,\ldots,4,0]$.
\end{lemma}
Figure~\ref{fig:layerShortcut}(a) illustrates Lemma~\ref{lem:spiral}.
Let $\ell_\rightarrow$ be a right-oriented link on a wire $w$ and w.l.o.g. let $s > 0$ be the spirality of $w$.
Further, let $L$ be a vertical line through $\ell_\rightarrow$ and $S$ a subsequence from $w$ with the properties guaranteed by Lemma~\ref{lem:spiral}.
Finally, let $\ell^i\in S$ be the unique link with spirality $0\le i\le s$ in $S$.
We define the \emph{$i$-core} for $S$ (for $4\le i \le s$ and $i\pmod4=0$) as the region enclosed by the wire $w$ from the intersection between $\ell^{i-4}$ and $L$ to the intersection between $\ell^i$ and $L$ and the straight line segment along $L$ connecting them (see Fig~\ref{fig:layerShortcut}(b)).
We define the \emph{$i$-layer} for $S$ (for $4\le i \le s-4$ and $i\pmod4=0$) as the difference of the $i$-core and the $(i+4)$-core (see Fig~\ref{fig:layerShortcut}(c)).

\begin{figure}[tb]
\centering
\includegraphics[scale=0.85]{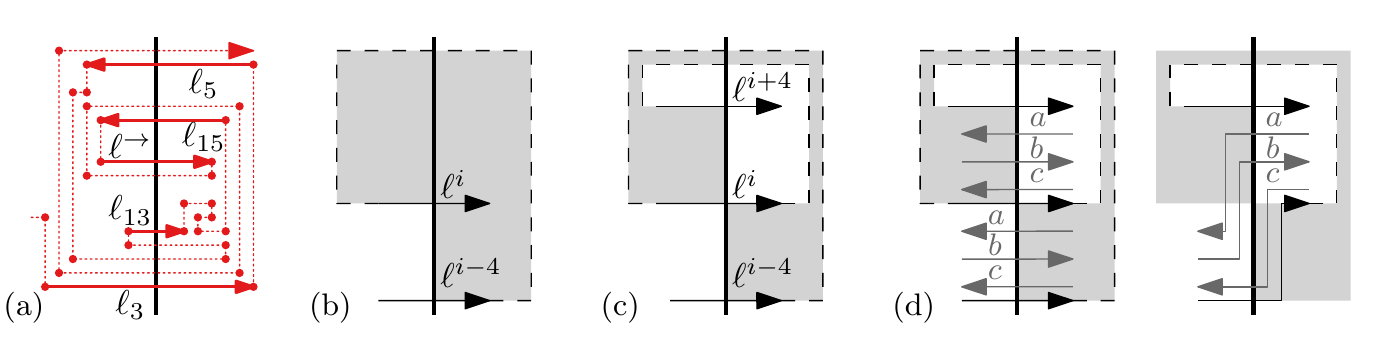}
\caption{\subcap{a} Lemma~\ref{lem:spiral} for a link $\ell_\rightarrow$ and sequence $S=(\ell_3,\ell_5,\ell_{13},\ell{15},\ell_\rightarrow)$. \subcap{b} The $i$-core of a spiral for a link $\ell^i\in S$ (gray). \subcap{c} The $i$-layer of the spiral (gray). \subcap{d} A layer cannot only contain wires as then we can shorten all wires.}
\label{fig:layerShortcut}
\end{figure}

\begin{lemma}\label{lem:horizontalWires}
An equivalent set of lr-wires with spirality $O(n)$ exists.
\end{lemma}
\begin{proof}
\emph{(Sketch)} We prove by induction that we can add a new lr-wire with spirality $O(n)$. If a wire $w$ has $\omega(n)$ layers, then we can argue via shortcuts (see Fig.~\ref{fig:layerShortcut}(d)) that $w$ was not shortest with respect to previously inserted wires.
\end{proof}

\begin{lemma}\label{lem:completeWires}
An equivalent set of wires with spirality $O(n)$ exists.
\end{lemma}
\begin{proof}
\emph{(Sketch)}
By Lemma~\ref{lem:horizontalWires} we can insert all lr-wires with spirality $O(n)$.
By Lemma~2 from~\cite{vanGoethemVerbeek2018} the spirality of intersecting links is the same.
Apply Lemma~\ref{lem:horizontalWires} for the tb-wires in the regions between the intersections with lr-wires.
\end{proof}

\begin{theorem} \label{thm:linear}
Let $\Gamma_I$ and $\Gamma_O$ be two unified planar orthogonal drawings of a (disconnected) graph $G$. We can morph $\Gamma_I$ into $\Gamma_O$ using $\Theta(n)$ linear morphs while maintaining planarity and orthogonality.
\end{theorem}
\begin{proof}
By Lemma~\ref{lem:completeWires} an equivalent set of wires with spirality $s=O(n)$ exists. By Theorem~8~from~\cite{vanGoethemVerbeek2018} we can thus morph the drawings into each other using $O(s) = O(n)$ linear morphs.
The lower bound of $\Omega(n)$ follows from~\cite{Alamdari2016}.
\end{proof}

\section{Combining intermediate linear morphs}\label{sec:noConstant}

The proof of Theorem~\ref{thm:linear} implies a morph between two unified planar orthogonal drawings $\Gamma_I$ and $\Gamma_O$ exists using $O(s)$ linear morphs, where $s$ is the spirality of $\Gamma_I$. In this section we show how to combine consecutive linear morphs into a total number of only $s$ linear morphs, while maintaining planarity and orthogonality. 

The morphs we describe can be encoded by a sequence of drawings, starting with $\Gamma_I$ and ending with $\Gamma_O$, such that every consecutive pair of drawings can be linearly interpolated while maintaining planarity and orthogonality.
For notational convenience let $\Gamma_i \morph \Gamma_j$ indicate that $\Gamma_i$ occurs before $\Gamma_j$ during the morph and $\Gamma_i \morpheq \Gamma_j$ that $\Gamma_i \morph \Gamma_j$ or $\Gamma_i = \Gamma_j$.

Let an \emph{iteration} of the original morph consist of all linear slides that jointly reduce spirality by one.
Let the first drawing of iteration $s$ be the first drawing in the original morph with spirality $s$ and the last drawing be the first drawing with spirality $s-1$.
Consecutive iterations overlap in exactly one drawing. These drawings in the overlap of iterations are the intermediate steps of the final morph.
Within this section let $\Gamma_I\morpheq \Gamma_a \morph \Gamma_b \morpheq \Gamma_O$, where $\Gamma_a$ is the first drawing with spirality $s$ and $\Gamma_b$ is the first drawing with spirality $s-1$. 


\subsection{Staircases}
Consider two distinct vertices $v$ and $w$ of the drawing.
Define an \emph{$x$-inversion} (\emph{$y$-inversion}) of $v$ and $w$ between $\Gamma_a$ and $\Gamma_b$ when the sign ($+$,$-$,$0$) of $v.x - w.x$ ($v.y - x.y$) differs in $\Gamma_a$ and $\Gamma_b$.
We say two vertices are \emph{$x$-inverted} (\emph{$y$-inverted}), or simply \emph{inverted}.
Two vertices $v$ and $w$ are \emph{separated} in a drawing by a link $\ell$ when they are both in the vertical (horizontal) strip spanned by $\ell$, and $v$ and $w$ are on opposite sides of $\ell$.

\begin{lemma}\label{lem:order}
Two vertices $v$ and $w$ can be inverted by a zigzag-removing slide along link $\ell$, if and only if $v$ and $w$ are separated by $\ell$.
\end{lemma}

\begin{wrapfigure}[6]{r}{2.2cm}
\raggedleft
\vspace{-.8cm}
\includegraphics[scale=0.85]{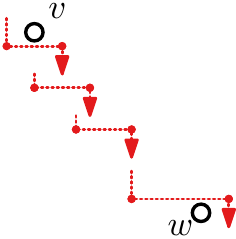}
\label{fig:staircase}
\end{wrapfigure}
\noindent
A \emph{downward staircase} is a sequence of horizontal links where: $(1)$ the left-endpoints are $x$-monotone increasing and $y$-monotone decreasing, $(2)$ the projection on the $x$-axis is overlapping or touching for a pair if and only if they are consecutive in the sequence, and $(3)$ all links have positive spirality.
Two vertices $v$ and $w$ are \emph{separated} by a downward staircase if $v$ is in the vertical strip spanned by the first link of the staircase and above it and $w$ is in the vertical strip spanned by the last link and below it.
Similar concepts can be defined for upwards staircases and for vertical links.

\begin{figure}[b]
\begin{minipage}[b]{0.45\linewidth}
\centering
\includegraphics[scale=0.85]{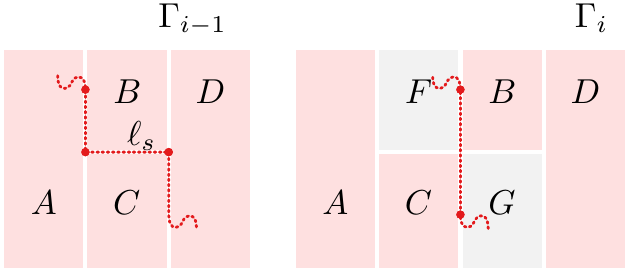}
\caption{Regions surrounding $\ell_s$ in $\Gamma_{i-1}$ and the matching regions in $\Gamma_i$.}
\label{fig:negative_small}
\end{minipage}
\hspace{0.5cm}
\begin{minipage}[b]{0.45\linewidth}
\centering
\includegraphics[scale=0.85]{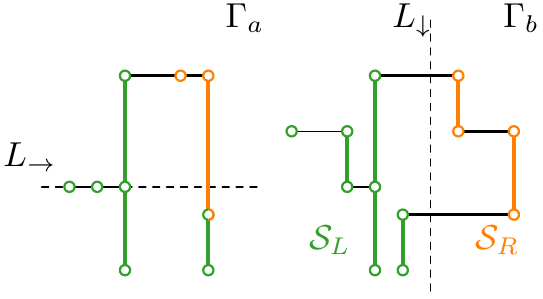}
\caption{Sets $\mathcal{S}_L$ and $\mathcal{S}_R$ in $\Gamma_a$ and $\Gamma_b$.}
\label{fig:intervals2}
\end{minipage}
\end{figure}

\begin{lemma}\label{lem:staircase}
Two vertices $v$ and $w$ that are $x$-inverted ($y$-inverted) first during a morph from $\Gamma_a$ to $\Gamma_b$, are separated by a horizontal (vertical) staircase of maximum spirality links in $\Gamma_a$.
\end{lemma}
\begin{proof}
\emph{(Sketch)}
Assume w.l.o.g. that only one inversion occurs and it occurs from $\Gamma_{b-1}$ to $\Gamma_b$.
By Lemma~\ref{lem:order}, $v$ and $w$ are separated by a link $\ell$ in $\Gamma_{b-1}$. Link $\ell$ must have maximum absolute spirality as it was selected for the morph. We now prove inductively that a staircase exists in all drawings from $\Gamma_a$ to $\Gamma_{b-1}$ by ``moving backwards'' through the morph.
To this end we define four rectangular regions $A,B,C,D$ surrounding $\ell_s$ in $\Gamma_{i-1}$ (see Fig.~\ref{fig:negative_small}). During the linear slide from $\Gamma_{i-1}$ to $\Gamma_i$ two new regions $F$ and $G$ are created, which cannot contain vertices. Using these rectangular regions and a case distinction on the type of linear slide, we can argue inductively that a staircase separating $v$ and $w$ must also exist in $\Gamma_{i-1}$.
\end{proof}

\subsection{Inversions}
We show that every pair of vertices is inverted along at most one axis during the morph from $\Gamma_a$ to $\Gamma_b$.
We then prove that $\Gamma_a$ has spirality one relative to $\Gamma_b$.

\begin{lemma}\label{lem:xORy}
Two vertices $v$ and $w$ can be inverted along only one axis during the morph from $\Gamma_a$ to $\Gamma_b$.
\end{lemma}

\begin{lemma}\label{lem:equivalent}
Each vertical (horizontal) line in $\Gamma_b$ not crossing a vertex, can be matched to a $y$- ($x$-)monotone wire in $\Gamma_a$.
\end{lemma}
\begin{proof}
\emph{(Sketch)}
Consider a vertical line $L_\downarrow$ in $\Gamma_b$ not intersecting any vertex.
Line $L_\downarrow$ partitions the set of vertices and vertical edges in $\Gamma_b$ into two subsets $\mathcal{S}_L$ and $\mathcal{S}_R$.
Consider a horizontal line $L_\rightarrow$ in $\Gamma_a$ and consider the maximal intervals formed along it by elements from the same set $\mathcal{S}_L$ or $\mathcal{S}_R$ (see Fig.~\ref{fig:intervals2}).
Set $\mathcal{S}_L$ and $\mathcal{S}_R$ form exactly two maximal intervals along $L_\rightarrow$.
Thus a $y$-monotone line exists correctly splitting $\mathcal{S}_L$ and $\mathcal{S}_R$.
We can show that this $y$-monotone line must intersect horizontal edges in the correct order as well.
\end{proof}

\begin{lemma}\label{lem:spirOne}
Drawing $\Gamma_a$ has spirality one relative to $\Gamma_b$.
\end{lemma}

\subsection{Single linear morph}
We now show that any two planar orthogonal drawings $\Gamma_i$ and $\Gamma_j$, where $\Gamma_i$ has spirality one relative to $\Gamma_j$, can be morphed into each other using a single linear morph while maintaining planarity.
Two drawings are \emph{shape-equivalent} 
if for each edge the sequence of left and right turns is identical and the orientation of the initial segment is identical in both drawings.
We say two drawings are \emph{degenerate shape-equivalent} if edges may contain zero-length segments but an assignment of orientations to the segments exists that is consistent with both drawings.
Two (degenerate) shape-equivalent drawings are per definition also unified.
We can make $\Gamma_a$ degenerate shape-equivalent to $\Gamma_b$ by adding zero-length edges whenever maximum absolute spirality links in $\Gamma_a$ cross an edge.
We say two \emph{points} $p$ and $q$ on the drawing are \emph{split} by a wire when $p$ and $q$ lie on different sides of the wire.

\begin{lemma}\label{lem:singleMorph}
Let $\Gamma_I$ and $\Gamma_O$ be two degenerate shape-equivalent drawings, where $\Gamma_I$ has spirality one.
There exists a single linear morph from $\Gamma_I$ to $\Gamma_O$ that maintains planarity and orthogonality.
\end{lemma}
\begin{proof}
\emph{(Sketch)} The partition of the drawing by all wires defines \emph{cells}: regions of the plane not split by any wire. For each cell containing at least one bend or vertex, we can linearly interpolate all vertices and bends in $\Gamma_I$ to the unique vertex or bend location in $\Gamma_O$.
This directly defines a linear morph between $\Gamma_I$ and $\Gamma_O$.
To argue planarity of this morph, we assume for contradiction that there exist two points $p$ and $q$ on an edge or vertex of the drawing that coincide during the morph (excluding $\Gamma_I$ and $\Gamma_O$).
Then $p$ and $q$ must be $x$- and $y$-inverted in $\Gamma_O$ compared to $\Gamma_I$ and there must be two vertices $r$ and $s$ that are $x$- and $y$-inverted and split by at least a tb-wire and a lr-wire.
As the lr-wire and the tb-wire are monotone they cross at least three times (see Fig.~\ref{fig:monotone}). Contradiction.
\end{proof}

\begin{figure}[t]
\centering
\includegraphics[scale=0.85]{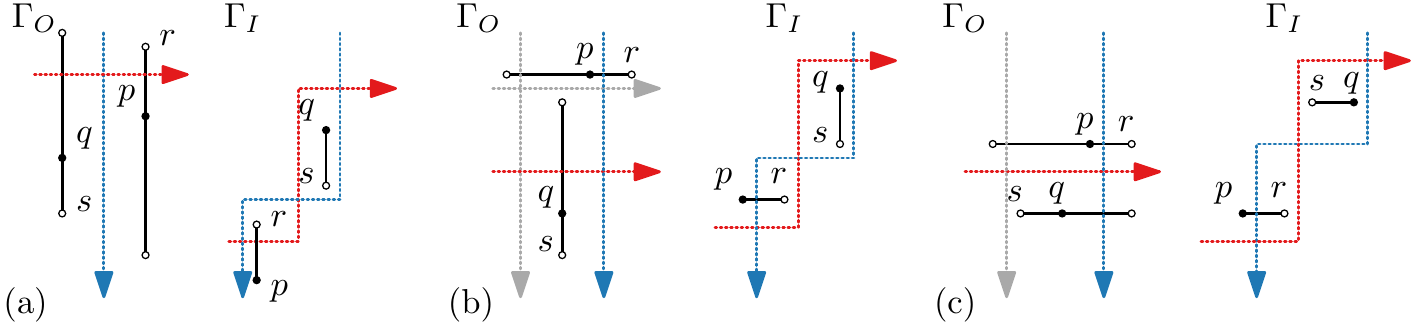}
\caption{\subcap{a} Two points $p$ and $q$ on vertical segments of the drawing that are inverted along both axes imply wires in $\Gamma_I$ that are not equivalent to $\Gamma_O$.  \subcap{b} Points $p$ and $q$ on a horizontal and vertical segment. \subcap{c} Points $p$ and $q$ on horizontal segments.}
\label{fig:monotone}
\end{figure}


\begin{theorem}\label{thm:sMorphs}
Let $\Gamma_I$ and $\Gamma_O$ be two unified planar orthogonal drawings of a (disconnected) graph $G$, where $\Gamma_I$ has spirality $s$. We can morph $\Gamma_I$ into $\Gamma_O$ using exactly $s$ linear morphs while maintaining planarity and orthogonality.
\end{theorem}


\section{Linear complexity of intermediate drawings}\label{sec:linearComplexity}

Van Goethem and Verbeek~\cite{vanGoethemVerbeek2018} describe rerouting and a simplification operations that reduce the complexity of intermediate drawings to $O(n)$. These operations are not compatible with the batched linear morphs we described in Section~\ref{sec:noConstant}. Below we show how to adapt these operations to the batched setting. These adaptations come at the cost of a single additional linear morph.

\subsection{Rerouting}\label{sub:rerouting}

To avoid that the linear morphs introduce too many bends in a single iteration of the morph, we show how to route the wires such that only $O(n)$ complexity is added to the drawing in each iteration. The initial rerouting of the wires in $\Gamma_I$ increases the maximum spirality by one, but it prevents any increase of spirality during the morph.
Thus, using Theorem~\ref{thm:sMorphs}, $s+1$ morphs are sufficient to morph two equivalent drawings into each other while maintaining planarity and keeping complexity of the intermediate drawings to $O(n^2)$. 

We reroute the wires in $W_\downarrow$ and $W_\rightarrow$ as follows.
Consider an edge $e$ that is crossed by at least two wires in $\Gamma_I$.
By Lemma~9 from~\cite{vanGoethemVerbeek2018} all crossing links have the same spirality.
Assume w.l.o.g. that this spirality is positive, otherwise mirror the rotations and replace right by left. 
Let $\varepsilon$ be a small distance such that the $\varepsilon$-band above $e$ is empty except for the links crossing $e$ and that there is more than a $2\varepsilon$ distance between the right-most crossing link and the right-endpoint of $e$ (see Fig.~\ref{fig:rerouting_new2}(a)).

We insert an \emph{$s$-windmill} of all crossing wires within the $\varepsilon$-band above $e$ by rerouting the wires as follows. First disconnect all crossing links within the $\varepsilon$-band above $e$. Then reroute all wires in a parallel bundle to the right, beyond the right-most wire $w_r$ crossing $e$. Now we spiral the bundle using right turns until the spirality of the links reaches zero. Next we unwind the bundle again within the spiral. Finally we reconnect the wires by routing back parallel to $e$ to maintain the original crossing points (see Fig.~\ref{fig:rerouting_new2}(b)). This rerouting can be executed without introducing crossings between the wires. It does increase the spirality of the drawing by one.

\begin{figure}[tb]
\centering
\includegraphics[scale=0.85]{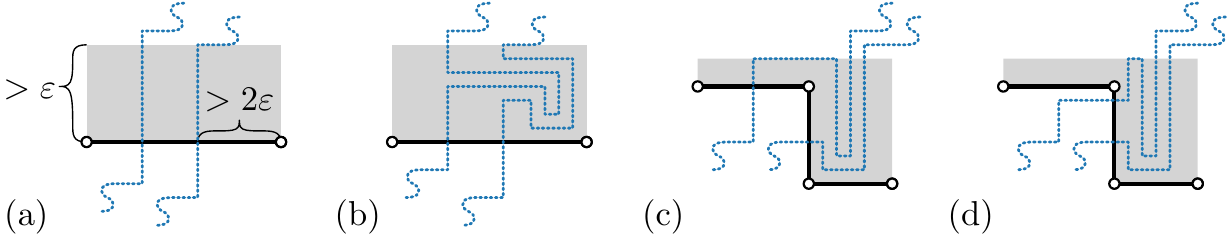}
\caption{\subcap{a} An $\varepsilon$-band adjacent to the edge. \subcap{b} Inserting an $s$-windmill. \subcap{c-d} Reroute wires after linear slide without introducing new crossings.}
\label{fig:rerouting_new2}
\end{figure}

We now change each iteration as follows. Consider a horizontal edge $e$ crossed by $k>1$ links of maximum absolute spirality $s$ (assuming $s > 0$) at the start of the iteration. Instead of performing a linear slide on all crossing links, we perform a single linear slide only on the rightmost crossing link. This slide creates a new vertical segment (see Fig.~\ref{fig:rerouting_new2}(c)). Thanks to the introduction of the $s$-windmill, we can easily reroute the other crossing wires to intersect the new vertical segment instead of the horizontal segment without introducing other crossings (see Fig.~\ref{fig:rerouting_new2}(d)). The newly created crossing links must have spirality $s-1$ as all links crossing the same segment have the same spirality (Lemma~9~from~\cite{vanGoethemVerbeek2018}). We can reduce all remaining spirality $s$ links without introducing additional complexity in the drawing.
\begin{lemma}\label{lem:bandInvariant}
At the start of iteration $i$ of the morph, all wires crossing an edge $e$ with links of spirality $i$ form an $i$-windmill in an empty $\varepsilon$-band next to $e$.
\end{lemma}
\begin{lemma}\label{lem:spirOneRerouting}
Let $\Gamma_s$ be the first drawing of an iteration and $\Gamma^r_{s-1}$ the rerouted last drawing.
The spirality of $\Gamma_s$ relative to $\Gamma^r_{s-1}$ is one.
\end{lemma}
\begin{proof}
\emph{(Sketch)} We can argue that rerouting wires does not eliminate staircases.
A link that is rerouted may have been part of a staircase, but the new links replacing it do not break any staircase properties.
As rerouting links maintains staircases, Lemmata~\ref{lem:staircase}-\ref{lem:spirOne} still apply.
\end{proof}
Drawing $\Gamma^r_{s-1}$ compared to $\Gamma_s$ contains two additional bends in each edge crossed by maximum absolute spirality links in $\Gamma_s$.
We can make $\Gamma_s$ and $\Gamma^r_{s-1}$ degenerate shape-equivalent by inserting an additional zero-length segment at the right-most (left-most for negative spirality) crossing link for each edge crossed by maximum absolute spirality links.
By Lemmata~\ref{lem:singleMorph} and \ref{lem:spirOneRerouting} we can morph the resulting $\Gamma_s$ into $\Gamma^r_{s-1}$ in a single linear morph while maintaining planarity.

As, independently of how many wires are crossing it, each edge only introduces two new bends, complexity increases by $O(n)$ during each iteration.
Thus the overall complexity is $O(s\cdot n)$.
We conclude that we can morph two drawings $\Gamma_I$ and $\Gamma_O$, where $\Gamma_I$ has spirality $s$, into each other using $s+1$ linear morphs while maintaining planarity and $O(s\cdot n)$ complexity of the drawing. 

\subsection{Simplification}\label{sub:simplification}
By using rerouting we can ensure that the complexity of the drawing increases by at most $O(n)$ in every iteration, but its complexity may still grow to $O(n^2)$ over $O(n)$ iterations. In this section we show how to simplify the intermediate drawings to ensure that the complexity after each iteration is $O(n)$. 

We again consider a single iteration starting with $\Gamma_s$ and ending with $\Gamma_{s-1}$.
Using rerouting we can find an alternative final drawing $\Gamma^r_{s-1}$ that also maintains planarity. We now introduce a \emph{redraw} step that further simplifies $\Gamma^r_{s-1}$ into a straight-line drawing $\Gamma'_{s-1}$ such that a linear morph from $\Gamma_s$ to $\Gamma'_{s-1}$ still maintains planarity. The redraw step works as follows.

For each vertex $v$ in $\Gamma^r_{s-1}$, consider a $6\varepsilon$-sized square box surrounding $v$ that contains only $v$ and a $3\varepsilon$-part of each outgoing edge from $v$. If an incident edge $e$ is crossed by a maximum absolute spirality link in $\Gamma_s$, then we reroute $e$ inside the $6\varepsilon$-box around $v$.
Specifically, for an edge $e$ leaving $v$ rightwards, we reroute $e$ within the $6\varepsilon$-box using the coordinates ($v, v+(0,-\varepsilon), v+(2\varepsilon,-\varepsilon), v+(2\varepsilon,0), v+(3\varepsilon,0)$) (see Fig.~\ref{fig:rerouteEdgesLocally}(a)).
Analogous rerouting can be done for edges leaving $v$ in other directions.
For an edge crossed by a negative spirality link invert the left and right turns.

\begin{figure}[tb]
\centering
\includegraphics[scale=0.85]{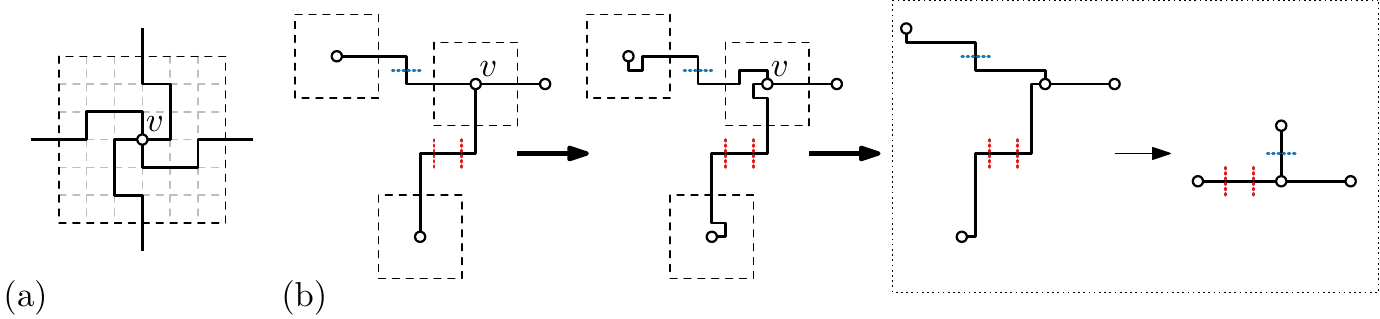}
\caption{\subcap{a} A $6\varepsilon$-box surrounding a vertex $v$ (dashed) with four redrawn edges. \subcap{b} Original drawing, rerouted drawing, and straightening the drawing.}
\label{fig:rerouteEdgesLocally}
\end{figure}

\begin{lemma}\label{lem:SafeRedraw}
We can redraw all edges in $\Gamma_{s-1}$ that were crossed by a maximum-spirality link in $\Gamma_s$ within $6\varepsilon$-boxes while maintaining planarity of the drawing.
\end{lemma}

\begin{proof}
\emph{(Sketch)}
We can establish a relation between the spiralities of two segments incident at the same vertex. Using this relation we can argue that, after redrawing, no two edges leave a vertex in the same direction. As a result, there are no planarity violations within the $6\varepsilon$-boxes around vertices.
\end{proof}

\begin{lemma}\label{lem:straightLineDrawing}
If $\Gamma_s$ is a straight-line drawing with spirality $s>0$ then there exists a straight-line drawing $\Gamma'_{s-1}$ with spirality $s-1$.
\end{lemma}
\begin{proof} \emph{(Sketch)}
Let $\Gamma^r_{s-1}$ be the drawing obtained by applying rerouting to the last drawing of iteration $s$. 
Consider an edge $e$ crossed by maximum absolute spirality links in $\Gamma_s$.
Edge $e$ has three segments in $\Gamma^r_{s-1}$ due to the two introduced bends.
The first and last segment do not cross any wires.
We can apply the redraw step to $e$, resulting in three more segments at the start and end of $e$.
Finally we eliminate all additional segments of $e$ by performing zigzag-eliminating slides on these segments (see Fig.~\ref{fig:rerouteEdgesLocally}(b)). 
\end{proof}

\begin{lemma}\label{lem:straightSpirOne}
The spirality of $\Gamma'_{s-1}$ relative to $\Gamma_{s}$ is one.
\end{lemma}

\begin{proof} \emph{(Sketch)}
Let the \emph{main wire set} be the set of wires used to compute the morph including rerouting from $\Gamma_s$ to $\Gamma^r_{s-1}$.
Consider a \emph{reference wire grid} that is a straight-line wire grid in $\Gamma_s$.
Using Lemmata~\ref{lem:equivalent},~\ref{lem:spirOne}, and~\ref{lem:spirOneRerouting} but swapping the roles of $\Gamma_a$ and $\Gamma_b$, we obtain the result that there is an equivalent monotone set of wires in $\Gamma^r_{s-1}$ matching the reference grid in $\Gamma_s$.
Thus the spirality of $\Gamma_s$ relative to $\Gamma^r_{s-1}$ is one.

When straightening $\Gamma^r_{s-1}$ to $\Gamma'_{s-1}$ only zigzag-removing slides are performed on segments not crossed by a wire from the main wire set.
As such a segment was not crossed by a wire from the main wire set, the orientation of the segment is unchanged in $\Gamma^r_{s-1}$.
Specifically, any link of a wire from the reference wire grid that crosses such a segment must have spirality zero.
When straightening $\Gamma^r_{s-1}$ to $\Gamma'_{s-1}$ the zigzag-removing slides may insert additional bends in these reference wires, but the wires will remain monotone.
\end{proof}

\noindent
We can make $\Gamma_s$ degenerate shape-equivalent to $\Gamma'_{s-1}$ as follows.
For each edge $e$ crossed by maximum absolute spirality links, we split $e$ at the crossing with the right-most (or left-most if the links have negative spirality) crossing link and insert a zero-length segment.
Furthermore, we add three zero-length segments at the endpoint of each such edge $e$ coincident with the respective endpoint.

\begin{theorem}\label{thm:spirOneStraight}
Let $\Gamma_I$ and $\Gamma_O$ be two equivalent drawings of a (disconnected) graph $G$, where $\Gamma_I$ has spirality $s$. We can morph $\Gamma_I$ into $\Gamma_O$ using $s+1$ linear morphs while maintaining planarity, orthogonality, and linear complexity of the drawing during the morph.
\end{theorem}

\mypar{Acknowledgements}
Bettina Speckmann and Kevin Verbeek are supported by the Netherlands Organisation for Scientific Research (NWO) under project no.~639.023.208 (B.S.) and no.~639.021.541 (K.V.). We want to thank the anonymous reviewers for their extensive feedback.

%
\newpage
\bibliographystyle{splncs04}
\bibliography{bibliographyMorphing}

\newpage
\appendix
\section{Omitted proofs}
\bigskip\noindent {\bf Theorem~\ref{thm:sMorphs}.\ }
\emph{Let $\Gamma_I$ and $\Gamma_O$ be two unified planar orthogonal drawings of a (disconnected) graph $G$, where $\Gamma_I$ has spirality $s$. We can morph $\Gamma_I$ into $\Gamma_O$ using exactly $s$ linear morphs while maintaining planarity and orthogonality.
}

\begin{proof}
By Lemma~\ref{lem:spirOne} for every iteration the initial drawing $\Gamma_s$ has spirality one relative to the first drawing $\Gamma_{s-1}$ of the next iteration.
We can make $\Gamma_s$ and $\Gamma_{s-1}$ degenerate shape-equivalent by adding a zero-length edge at the crossing of each edge with each maximum absolute spirality link in $\Gamma_s$.
By Lemma~\ref{lem:singleMorph} we can reduce the spirality of the drawing by one with a single linear morph.
\end{proof}

\bigskip\noindent {\bf Theorem~\ref{thm:spirOneStraight}.\ }
\emph{Let $\Gamma_I$ and $\Gamma_O$ be two equivalent drawings of a (disconnected) graph $G$, where $\Gamma_I$ has spirality $s$. We can morph $\Gamma_I$ into $\Gamma_O$ using $s+1$ linear morphs while maintaining planarity, orthogonality, and linear-complexity of the drawing during the morph.}

\begin{proof}
Use induction to prove there exists a straight-line drawing $\Gamma_k$ of spirality $0\le k \le s+1$ such that there exists a planar morph from $\Gamma_I$ to $\Gamma_k$ that is comprised of $(s+1)-k$ linear morphs.

For the base case take a set of wires in $\Gamma_I$ with spirality $s$.
Reroute the wires to insert windmills next to all crossed edges, thereby increasing spirality to $s+1$.
By assumption the input was a straight-line drawing and trivially $\Gamma_I=\Gamma_{s+1}$ so no linear morph is required.

For the step let $0<i\le s+1$ and let $\Gamma_i$ have spirality $i$.
By hypothesis $\Gamma_i$ is a straight-line drawing and a morph comprised of $(s+1)-i$ linear morphs exists starting at $\Gamma_I$ and ending at $\Gamma_i$.
Compute a morph from $\Gamma_i$ to $\Gamma_O$ as in~\cite{vanGoethemVerbeek2018}.
Moreover, apply the rerouting as discussed in Section~\ref{sec:noConstant} to avoid introducing excess complexity.
Let $\Gamma^r_{i-1}$ be the first drawing with spirality $i-1$ in this computed morph.
Redraw $\Gamma^r_{i-1}$ by inserting additional intersection-free segments and then straighten the resulting drawing to a straight-line drawing $\Gamma'_{i-1}$ as described.
Also morph the set of wires from $\Gamma^r_{i-1}$ along in this process.

By Lemma~\ref{lem:straightSpirOne} drawing $\Gamma'_{i-1}$ has spirality one relative to $\Gamma_{i}$.
Make $\Gamma_i$ and $\Gamma'_{i-1}$ degenerate shape-equivalent by inserting additional vertices in $\Gamma_i$ as discussed in Section~\ref{sub:simplification}.
By Lemma~\ref{lem:singleMorph} we can linearly morph $\Gamma'_{i-1}$ into $\Gamma_{i}$ without violating planarity.
Specifically we can perform this morph in reverse to obtain a morph from $\Gamma_i$ to $\Gamma'_{i-1}$.

Drawing $\Gamma'_{i-1}$ has spirality $i-1$, is straight-line, and concatenating the morph from $\Gamma_I$ to $\Gamma_i$ with the single linear morph from $\Gamma_i$ to $\Gamma_{i-1}$ results in a morph comprised of $(s+1)-(i-1)$ linear morphs.

We can remove any coincident or virtual bends from $\Gamma'_{i-1}$ to maintain $O(n)$ complexity.
\end{proof}

\newpage
\bigskip\noindent {\bf Lemma~\ref{lem:spiral}.\ }
\emph{For each right-oriented link $\ell_\rightarrow$ of a wire $w\in W_\rightarrow$ with positive (negative) spirality $s$ there exists a vertical line $L$ and a subsequence of $\Omega(|s|)$ links of $w$ crossing $L$, such that the absolute spiralities of the links in sequence are $[0,2,4,\ldots,|s|-2,|s|]$, and when ordered top-to-bottom (bottom-to-top) along $L$ form the sequence $[2,6,10,\ldots, |s|-2,|s|,|s|-4,\ldots,4,0]$.}

\begin{proof}
Let $w[i]$ be partial wire consisting of links $\ell_1,\ldots,\ell_i$ of $w$.
Consider a vertical line $L$ through $\ell_\rightarrow$.
We find the desired subsequence $S$ of $w$ by constructing it starting at the back.
Assume for induction that the subsequence $S=(\ell_i, \ell_j, \ldots, \ell_\rightarrow)$ with $1\le k \le \frac{s}{2}$ links has been constructed.
In the base case this is simply $S=(\ell_\rightarrow)$.

Link $\ell_i$ has spirality $t=s-2(k-1)$.
Make a distinction on the orientation of $\ell_i$.
Assume $\ell_i$ is left-oriented, a similar argument holds when $\ell_i$ is right-oriented.
As $\ell_i$ is left-oriented, $t\pmod 4=2$.
By the hypothesis $\ell_i$ occurs before $\ell_\rightarrow$ (with spirality $s$) when ordered top-to-bottom along $L$.
Moreover as $\ell_i$ has the smallest spirality from all selected links in $S$ so far, $\ell_i$ must be the highest link in $S$ that crosses $L$.

Let $\ell_m$ be the link from $w[i]$ crossing $L$ directly below $\ell_i$.
We show $\ell_m$ crosses $L$ lower than any link from $S$.
If $\ell_m$ is the second link selected this is vacuously true, otherwise let $\ell_j$ be the link found in the step before $\ell_i$.
By the same argumentation link $\ell_j$, with spirality $t+2\pmod 4 = 0$, is the lowest link from $S$ crossing $L$.
As $i<j$, $w[i] \subset w[j]$ and thus link $\ell_m \in w[j]$.
As $\ell_i$ was the link from $w[j]$ crossing $L$ directly above $\ell_j$, $\ell_m$ must cross $L$ below $\ell_j$ (see Fig.~\ref{fig:spiralityChange}(a)).

Link $\ell_m$ is the first link from $w[i]$ crossing $L$ below $\ell_i$, thus $w[i]$ (and therefore $w[m]$) cannot cross through $L$ between $\ell_i$ and $\ell_m$.
Consequently, the sub-wire from $\ell_m$ to $\ell_i$ cannot enclose the origin of $\ell_m$.
This implies there is unique topological way to connect $\ell_m$ to $\ell_i$ and $\ell_m$ must have spirality $t-2$ or $t+4$, which is uniquely defined by the orientation of $\ell_m$ (see Fig.~\ref{fig:spiralityChange}(b)).
If the spirality of $\ell_m$ is $t-2$ then we add $\ell_m$ to $S$, otherwise we repeat the downwards search.
As spirality decreases in steps of two and the lowest link crossing $L$ has spirality at most $0$ (from Lemma~3~in~\cite{vanGoethemVerbeek2018}) a suitable link with spirality $t-2$ will be found.
\end{proof}

\begin{figure}[h]
\centering
\includegraphics[scale=0.85]{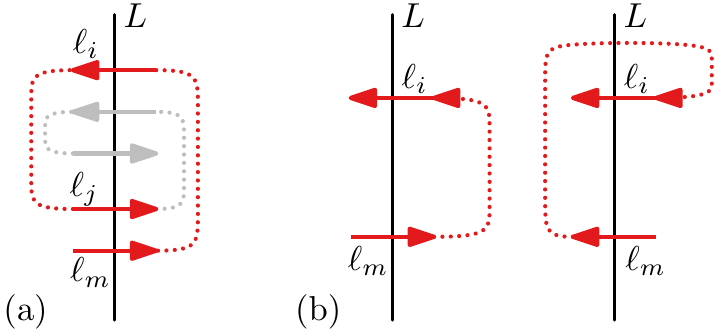}
\caption{\subcap{a} The first link $\ell_m$ from $w[i-1]$ crossing $L$ below $\ell_i$ must cross $L$ lower than the lowest link $\ell_j$ from $S$ crossing $L$. \subcap{b} There are exactly two configurations for link $\ell_m$ ($m<i$) crossing directly below left-oriented link $\ell_i$ (for partial wire $w[i-1]$). Either the spirality of $\ell_m$ is two smaller or four larger than that of $\ell_i$.}
\label{fig:spiralityChange}
\end{figure}

\newpage

\bigskip\noindent {\bf Lemma~\ref{lem:horizontalWires}.\ }
\emph{An equivalent set of lr-wires with spirality $O(n)$ exists.}

\begin{proof}
We prove the statement constructively by induction on the size of the constructed set.
Assume an equivalent set $\mathcal{S}$ of $k$ lr-wires with spirality $O(n)$ exists, where each wire is shortest with respect to the previously inserted wires.
In the base case $\mathcal{S}=\emptyset$.

For the step consider a wire $w_O\in W_\rightarrow$ that does not have a matching counterpart in $\mathcal{S}$.
Find the shortest lr-wire $w$ in $\Gamma_I$ that matches $w_O$ such that $\mathcal{S}\cup w$ is an equivalent set of size $k+1$.
Consider the right-oriented link $\ell\in w$ with maximum absolute spirality $s$.
By Lemma~\ref{lem:spiral} there exist a (topological) spiral around $\ell$ that has $\Omega(s)$ layers.
We bound the number of layers surrounding $\ell$ to $O(n)$.
Thus, it must also be that $s=O(n)$ and the spirality of $w$ is $O(n)$.
To achieve this we classify the layers of the spiral by their containment of the drawing, a layer (1) contains a vertex, or (2) is crossed by an edge, or (3) contains no part of the drawing (but may contain wires).

Clearly there are at most $O(n)$ layers that contain a vertex of the drawing.
If a layer contains no vertex, but does contain an edge then that edge must cross through the layer. 
As each edge is crossed at most once by $w$, such an edge must cross $L$ exactly once.
Each edge can only be involved in the layers left and right of the crossing with $L$ in this way.
Thus only $O(n)$ layers contain an edge but no vertex.
The remaining layers contain subsections of wires or are empty.
Trivially an empty layer cannot exists, as otherwise $w$ is not shortest.

We show that no layer exists that contains only wires (including parts of $w$ itself).
The boundary of a layer is formed by $w$ and two straight-line segments along $L$.
We refer to two parts of the boundary along $L$ as the \emph{gates} of the layer.

Assume for contradiction there is a layer $R$ that only contains subsections of wires (possibly $w$ itself).
Lr-wires do not cross and hence must enter and leave $R$ through the gates.
As $R$ contains no part of the drawing and the lr-wires are each shortest with respect to the previously inserted wires they cannot consecutively enter and leaves $R$ through the same gate.
Moreover, the order of the wires at both gates is identical.

Disconnect all lr-wires at the gates of $R$.
Also disconnect $w$ at the lower link adjacent to each gate.
Remove all disconnected components.
Reconnect the remaining parts locally along $L$ to ensure the remaining links of all wires are visited in the same order and no crossings occur (see Fig.~\ref{fig:layerShortcut}(d)).
All wires crossing $R$ have been shortened by this.
Contradiction, as all wires in the existing equivalent set as well as $w$ itself, were shortest with respect to the previously inserted wires.

We conclude that there can be at most $O(n)$ layers.
Therefore, the maximum spirality of any link and thus the newly introduced lr-wire is $O(n)$.
\end{proof}

\newpage

\bigskip\noindent {\bf Lemma~\ref{lem:completeWires}.\ }
\emph{An equivalent set of wires with spirality $O(n)$ exists.}

\begin{proof}
By Lemma~\ref{lem:horizontalWires} we can insert all lr-wires with spirality $O(n)$.
By Lemma~2 from~\cite{vanGoethemVerbeek2018} the spirality of intersecting links is the same.
Thus when a tb-wire intersects a lr-wire it has spirality $O(n)$ and we can consider the regions between these intersections individually.
Between two intersections no pair of tb-wires intersect.
Furthermore, there are no crossings with the pair of lr-wires at the border of the region.
We can consider the same proof as Lemma~\ref{lem:horizontalWires} where a wire may now either be a lr-wire or a tb-wire.
Thus, in the region between two lr-wires the spirality of each tb-wire increases (decreases) by at most $O(n)$ and is $O(n)$ again when crossing the second lr-wire.
It follows that the overall spirality of the tb-wires is also $O(n)$.
\end{proof}

\bigskip\noindent {\bf Lemma~\ref{lem:order}.\ }
\emph{Two vertices $v$ and $w$ can be inverted by a zigzag-removing slide along link $\ell$, if and only if $v$ and $w$ are separated by $\ell$.}
\begin{proof}
W.l.o.g. assume $\ell$ is vertical and the spirality is positive (see Fig.~\ref{fig:strip}).
Let $\mathcal{V}$ be the set of vertices moved by a zigzag-removing slide on $\ell$.
If $v,w\in \mathcal{V}$ or $v,w\not\in\mathcal{V}$ then $v,w$ are moved equally in the same direction and cannot be inverted.
Hence either $v \in \mathcal{V}$ or $w\in \mathcal{V}$; assume $v\in\mathcal{V}$.
All vertices in $\mathcal{V}$ move up by the length of $\ell$.
To be inverted we need that initially $w.y > v.y$, but also that $w \not\in\mathcal{V}$.
But then $v$ and $w$ must both be in the strip spanned by $\ell$ and they must be separated by $\ell$.
\end{proof}

\begin{figure}[h]
\centering
\includegraphics[scale=0.85]{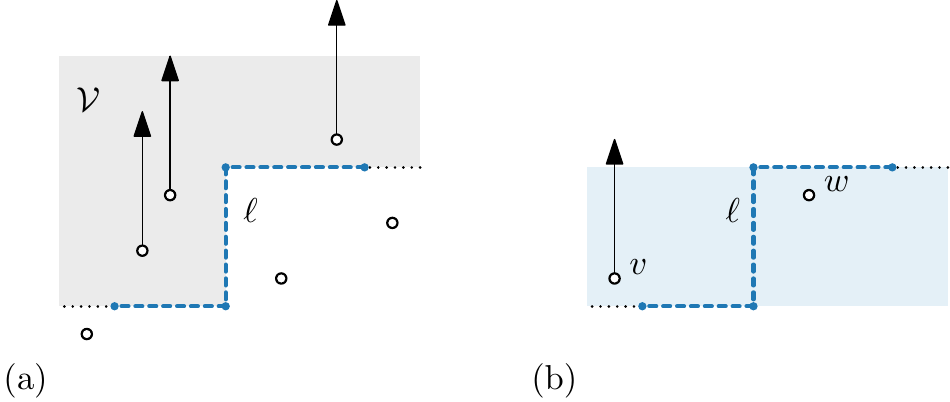}
\caption{\subcap{a} Motion of the vertices in $\mathcal{V}$ (gray) defined by the (horizontally extended) zigzag containing $\ell$. \subcap{b} To change the order of $v$ and $w$ along the $y$-axis, both must be in the horizontal strip defined by $\ell$ (blue) and separated by $\ell$.}
\label{fig:strip}
\end{figure}

\newpage
Let $\Gamma_I\morpheq \Gamma_a \morph \Gamma_b \morpheq \Gamma_O$, where $\Gamma_a,\Gamma_b$ are the first and last drawing of an iteration.
Specifically $\Gamma_a$ is the first drawing with spirality $s$ and $\Gamma_b$ is the first drawing with spirality $s-1$.

\bigskip\noindent {\bf Lemma~\ref{lem:staircase}.\ }
\emph{Two vertices $v$ and $w$ that are $x$-inverted ($y$-inverted) first during a morph from $\Gamma_a$ to $\Gamma_b$, are separated by a horizontal (vertical) staircase of maximum spirality links in $\Gamma_a$.}

\begin{figure}[t]
\centering
\includegraphics[scale=0.85]{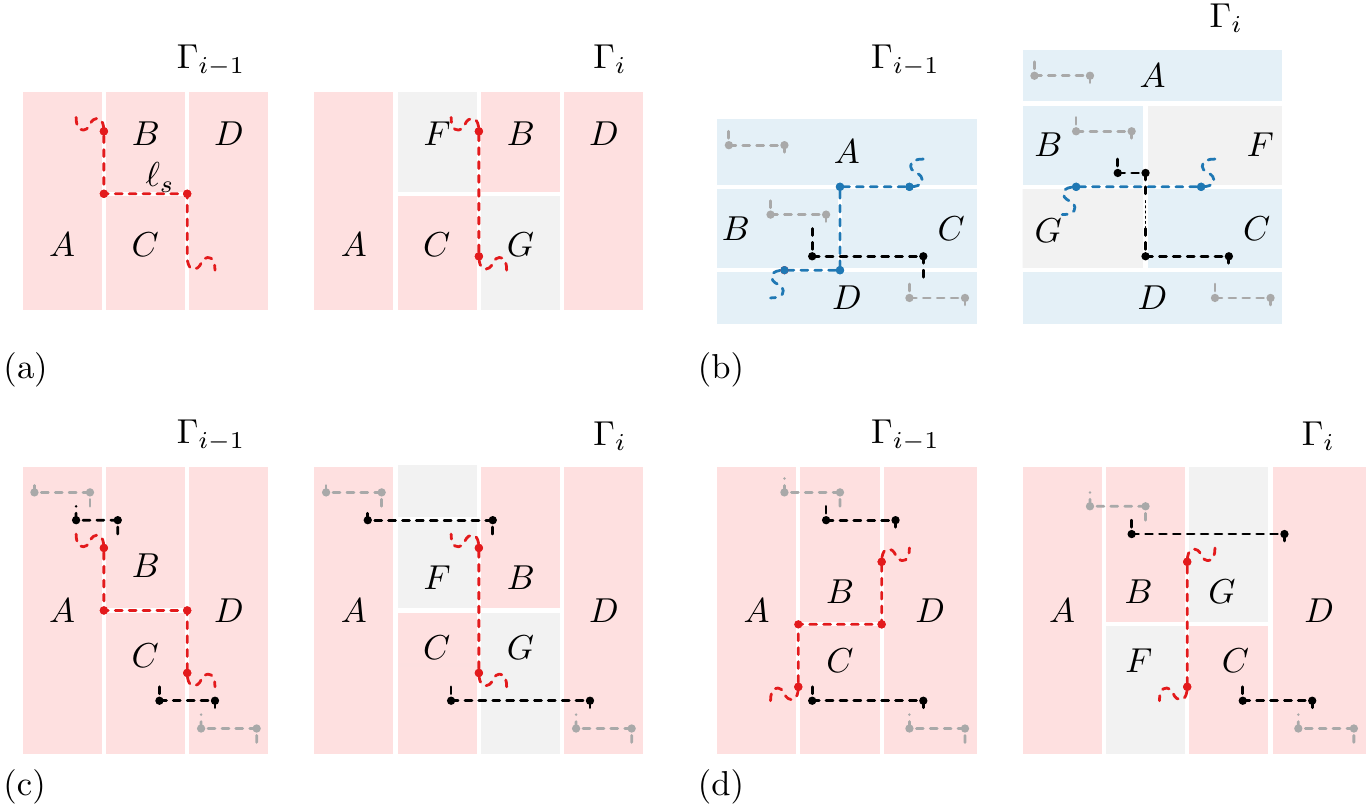}
\caption{\subcap{a} Regions surrounding $\ell_s$ in $\Gamma_{i-1}$ and the matching regions in $\Gamma_i$. Regions $A,B,C,D$ are maintained between the two drawings. \subcap{b} A vertical slide may merge two links from staircase $S$ in $\Gamma_{i-1}$ but the result is also a valid staircase. \subcap{c} If $S$ is split in $\Gamma_{i-1}$ then it can be extended by adding $\ell_i$. \subcap{d} If non-adjacent links from $S$ overlap in $\Gamma_{i-1}$ then we can select a subsequence from $S$ forming a staircase in $\Gamma_{i-1}$.}
\label{fig:negative}
\end{figure}

\begin{proof}
Assume w.l.o.g. that only one inversion occurs and it occurs from $\Gamma_{b-1}$ to $\Gamma_b$, otherwise consider the initial part of the morph.
Assume that $v.x<w.x$, $v.y > w.y$ in all drawings from $\Gamma_a$ to $\Gamma_{b-1}$ and $v.x>w.x$, $v.y>w.y$ in $\Gamma_b$.
We prove the claim for the intermediate drawings in backwards direction starting from $\Gamma_{b-1}$.
By Lemma~\ref{lem:order} $v$ and $w$ are separated by a single maximum absolute spirality link $\ell$ in $\Gamma_{b-1}$.
This trivially satisfies all constraints for a downwards staircase.
As $v.y > w.y$ and $v.x<w.x$ in $\Gamma_{b-1}$ and $v.x>w.x$ in $\Gamma_b$, link $\ell$ has positive spirality.

Assume the sequence $S$ composes a downwards staircase in $\Gamma_i$, where $\Gamma_{a} \morph \Gamma_i \morpheq \Gamma_{b-1}$.
For convenience of argument we consider $v$ and $w$ as zero-length horizontal links that are part of $S$.
We show a downwards staircase separating $v$ and $w$ also exists in $\Gamma_{i-1}$.
Let link $\ell_s$ be the link causing the linear slide from $\Gamma_{i-1}$ to $\Gamma_i$.
Define four rectangular regions $A,B,C,D$ surrounding $\ell_s$ that partition the plane in $\Gamma_{i-1}$ (see Fig.~\ref{fig:negative}(a)). 
During the linear slide from $\Gamma_{i-1}$ to $\Gamma_i$ all four regions are maintained intact.
Moreover, two new regions $F,G$ are present in $\Gamma_i$.
However, as regions $A,B,C,D$ are maintained intact and together contain all vertices, regions $F,G$ do not contain any vertices.

Assume $S$ is not also a downwards staircase in $\Gamma_{i-1}$ otherwise we are done.
To break any of the staircase properties the $x$- ($y$-)order of two endpoints of links in $S$ must change between $\Gamma_{i-1}$ and $\Gamma_i$
This can only occur if the respective endpoints are separated by $\ell_s$ (Lemma~\ref{lem:order}) and hence at least one link from $S$ must be overlap with region $B$ and one link must overlap with region $C$.
Let $S_1$ be the sub-staircase consisting of the links of $S$ upto the last link that intersects, or is contained in, region $B$.
Staircase $S_2$ consists of the remaining links and, specifically, the first link of $S_2$ must intersect or be contained in region $C$.
As the endpoints of the links in $S$ are monotone decreasing in $y$, region $C$ contains no links from $S_1$, and region $B$ contains no links from $S_2$.
Hence, as no pair of vertices in a sub-staircase is separated by $\ell_s$, all staircase properties are maintained for $S_1$ and $S_2$ separately between $\Gamma_{i-1}$ and $\Gamma_i$.
Let $\ell_1$ be the last link from $S_1$ and $\ell_2$ the first link from $S_2$.
Make a case distinction on the orientation of $\ell_s$.

Assume $\ell_s$ is vertical (see Fig.~\ref{fig:negative}(b)).
A vertical slide can falsify only the $y$-monotonicity of $S$.
As the $x$-projection of $\ell_1$ and $\ell_2$ touches (overlaps) in $\Gamma_i$ and regions $F,G$ contain no vertices $\ell_1$ ends at the right border of $B$ and $\ell_2$ starts at the left border of $C$.
Any vertical linear slide (degenerately) maintains the $y$-order on each vertical line.
Thus also in $\Gamma_{i-1}$ the right endpoint of $\ell_1$ is above (or equal with) the left endpoint of $\ell_2$.
In the boundary case $\ell_1$ and $\ell_2$ form a single link in $\Gamma_{i-1}$.
In each case $S$ forms of a valid staircase in $\Gamma_{i-1}$.

Assume $\ell_s$ is horizontal.
A horizontal slide can falsify the $x$-monotonicity or the overlap of links.
Consider the spirality of $\ell_s$.

First, assume $\ell_s$ has positive spirality (Fig~\ref{fig:negative}(c)).
As $S$ is a valid staircase in $\Gamma_i$ the projection on the $x$-axis of any two non-adjacent links is non-overlapping.
Then $\ell_1$ and $\ell_2$ cannot be fully contained in $B$ respectively $C$ as otherwise either $\ell_1$ is contained in $\ell_2$ in $\Gamma_i$ or vice versa, and therefore at least one pair of non-adjacent links must overlap.
Moreover, from $S$ only $\ell_1$ enters $B$ and only $\ell_2$ enters $C$.
In $\Gamma_{i-1}$ only the right endpoint of $\ell_1$ and the left endpoint of $\ell_2$ may be inverted.
If not, then $S$ is a downwards staircase in $\Gamma_{i-1}$ as well.
If so, then by Lemma~\ref{lem:order} $\ell_1$ and $\ell_2$ are separated by $\ell_s$.
Thus $(S_1,\ell_s,S_2)$ forms a downwards staircase in $\Gamma_{i-1}$.

Second, assume $\ell_s$ has negative spirality (Fig~\ref{fig:negative}(d)).
Then in $\Gamma_{i-1}$ the projection on the $x$-axis of non-adjacent links from $S$ may overlap.
As $v.x<w.x$ in both $\Gamma_{i-1}$ and $\Gamma_i$ at least some pair of links from $S_1$ and $S_2$ overlap.
Select a subsequence from $S$ satisfying all constraints by dropping links from the end of $S_1$ until only one link from $S_1$ overlaps a link from $S_2$.
Then drop links from the start of $S_2$ until no non-adjacent pair of links in $S$ overlap.
\end{proof}

\begin{note}
If two vertices $v,w$ are inverted along both axes during a morph from some drawing $\Gamma_a$ to some drawing $\Gamma_b$, then there exists a submorph where $v$ and $w$ are inverted exactly once along both axes.
\end{note}

\newpage

\bigskip\noindent {\bf Lemma~\ref{lem:xORy}.\ }
\emph{Two vertices $v$ and $w$ can be inverted along only one axis during the morph from $\Gamma_a$ to $\Gamma_b$.}

\begin{proof}
Assume for contradiction a pair of vertices $v,w$ exists that is inverted along both axes.
Assume they are inverted along both axis exactly once, otherwise consider the submorph where this is the case.
W.l.o.g. let $v.x < w.x$, $v.y<w.y$ in $\Gamma_a$, $v.x < w.x$, $v.y>w.y$ in all drawings from $\Gamma_{a+1}$ to $\Gamma_{b-1}$, and $v.x > w.x$, $v.y>w.y$ in $\Gamma_b$.

By Lemma~\ref{lem:staircase}, and the relative position of $v$ and $w$, there exists a downwards staircase separating $v$ and $w$ in $\Gamma_{a+1}$.
By Lemma~\ref{lem:order} and the inversion of $v,w$ from $\Gamma_a$ to $\Gamma_{a+1}$ there must be a vertical link with positive spirality separating $v,w$ in $\Gamma_a$.
We consider the regions $A,B,C,D$ surrounding the link in $\Gamma_a$ that causes the linear slide from $\Gamma_a$ to $\Gamma_{a+1}$ and the matching regions $F,G$ in $\Gamma_{a+1}$.
As regions $F$ and $G$ cannot contain any vertices, the staircase separating $v,w$ must have a link $\ell_1$ ending at the right boundary of $B$ and a link $\ell_2$ ending at the left boundary $C$ (see Fig.~\ref{fig:negative}(b)).
Specifically, the $y$-coordinate of $\ell_1$ is greater or equal than the $y$-coordinate of $\ell_2$ in $\Gamma_{a+1}$.
As any vertical slide, in particular the one from $\Gamma_a$ to $\Gamma_{a+1}$, maintains the order on a vertical line we must also have $v.y > w.y$ in $\Gamma_a$.
Contradiction.
\end{proof}

\bigskip\noindent {\bf Lemma~\ref{lem:equivalent}.\ }
\emph{Each vertical (horizontal) line in $\Gamma_b$ not crossing a vertex, can be matched to a $y$- ($x$-)monotone wire in $\Gamma_a$.}

\begin{figure}[tb]
\centering
\includegraphics[scale=0.85]{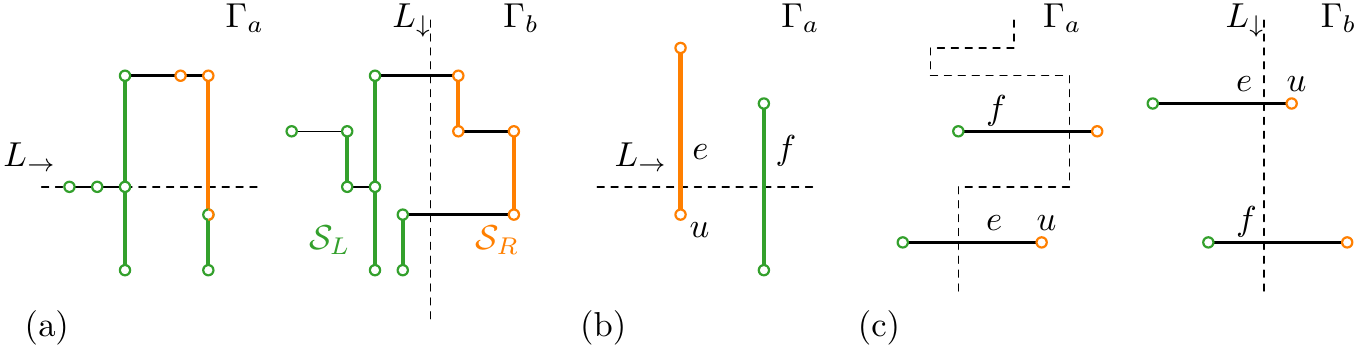}
\caption{\subcap{a} Sets $\mathcal{S}_L$ and $\mathcal{S}_R$ in $\Gamma_a$ and $\Gamma_b$. \subcap{b} A vertical edge $e\in\mathcal{S}_R$ cannot cross $L_\rightarrow$ left of a vertical edge $f\in\mathcal{S}_L$ as vertex $u$ must be $x$- and $y$-inverted with one of the endpoints of $f$ during the morph. \subcap{c} The $y$-monotone line cannot cross the edges in the wrong order as then vertex $u$ must be $x$- and $y$-inverted with an endpoint of $f$.}
\label{fig:intervals}
\end{figure}

\begin{proof}
Let $L_\downarrow$ be a vertical line in $\Gamma_b$ not intersecting any vertex.
Line $L_\downarrow$ partitions the set of vertices and vertical edges in $\Gamma_b$ into two subsets $\mathcal{S}_L,\mathcal{S}_R$.
We consider the shape of these sets in $\Gamma_a$.
Let $L_\rightarrow$ be a horizontal line in $\Gamma_a$ and consider the elements from $\mathcal{S}_L$ and $\mathcal{S}_R$ that intersect $L_\rightarrow$.
Order the elements left-to-right along $L_\rightarrow$ and form maximal subsets of consecutive elements from the same set (see Fig.~\ref{fig:intervals}(a)).
If vertices from $\Gamma_b$ coincide in $\Gamma_a$, then we order those vertices such that vertices from $\mathcal{S}_L$ come before vertices from $\mathcal{S}_R$.
To differentiate we refer to the maximal subsets formed as \emph{intervals}.

Assume for contradiction there are two consecutive intervals $I_R \subseteq \mathcal{S}_R$ and $I_{L}\subseteq \mathcal{S}_L$ when ordered left-to-right along $L_\rightarrow$.
Let $e\in I_R$ and $f\in I_{L}$. 
We have $e.x<f.x$ in $\Gamma_a$ and $e.x>f.x$ in $\Gamma_b$.
As $e$ and $f$ are $x$-inverted, they (or their endpoints) cannot be $y$-inverted between $\Gamma_a$ and $\Gamma_b$ (Lemma~\ref{lem:xORy}).
Assume $e,f$ are vertical edges in $\Gamma_a$ and an endpoint $u$ of $e$ is in the horizontal strip defined by $f$ (see Fig.~\ref{fig:intervals}(b)). The case where $e$ or $f$ are vertices is analogous.
As the morph is planar $u$ cannot move through $f$ while morphing to $\Gamma_b$.
But then $u$ changed in the $y$-order with at least one of the endpoints of $f$ during the morph.
Contradiction.

We conclude no interval composed of elements from $\mathcal{S}_L$ can come after an interval composed of elements from $\mathcal{S}_R$.
Thus there at most two maximal intervals on each horizontal line in $\Gamma_a$ and they are ordered with the elements from $\mathcal{S}_L$ occurring first.
It follows that a $y$-monotone wire must exist in $\Gamma_a$ that correctly partitions the vertices and vertical edges.

A $y$-monotone line intersects each horizontal edge at most once and, as it partitions the vertices correctly, it must intersect exactly the required horizontal edges.
It is left to prove that such a $y$-monotone wire also intersects the horizontal edges in the correct order.
Consider an arbitrary pair of horizontal edges $e,f$ that is intersected by $L_\downarrow$ in this order in $\Gamma_b$.
If $e,f$ have the same vertical order in $\Gamma_a$ then the claim trivially holds.
Assume the end-points of $e,f$ are $y$-inverted in $\Gamma_a$ (see Fig.~\ref{fig:intervals}(c)).
Then by Lemma~\ref{lem:xORy} the $x$-order of the end-points is the same in $\Gamma_a$ and $\Gamma_b$.
W.l.o.g. of assume an endpoint $u$ of $e$ is in the horizontal strip defined by $f$ in $\Gamma_a$ and $\Gamma_b$.
Vertex $u$ must have changed in the $x$-order with at least one of the endpoints of $f$ during the morph.
Contradiction.
\end{proof}

\newpage
\bigskip\noindent {\bf Lemma~\ref{lem:spirOne}.\ }
\emph{Drawing $\Gamma_a$ has spirality one relative to $\Gamma_b$.}

\begin{figure}[t]
\centering
\includegraphics[scale=0.85]{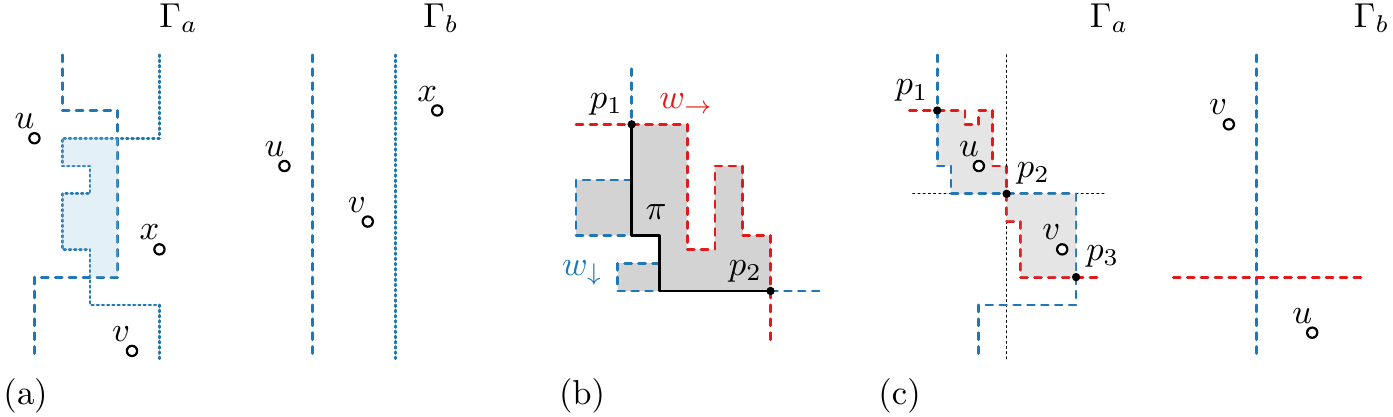}
\caption{\subcap{a} The number of crossings of two $y$-monotone tb-wires that cross at least twice can be reduced as the areas in between them (light blue) cannot contain any vertices. \subcap{b} Each enclosed region contains an $xy$-monotone path $\pi$. \subcap{c} An $x$-monotone lr-wire and a $y$-monotone tb-wire cannot cross three times.}
\label{fig:proper}
\end{figure}

\begin{proof}
To determine the spirality of $\Gamma_a$ relative to $\Gamma_b$ we need to show a monotone wire-grid exists in $\Gamma_a$ that is equivalent to a straight-line wire grid in $\Gamma_b$.
To this end, consider a new straight-line wire-grid in $\Gamma_b$.
By Lemma~\ref{lem:equivalent} we can find an equivalent $x$- ($y$-) monotone wire in $\Gamma_a$ for each straight wire in $\Gamma_b$.
All that is left to show is that there is also set of monotone wires in $\Gamma_a$ that together form an equivalent set to the wires in $\Gamma_b$.

Assume for contradiction such a set does not exist. 
Then for any set of monotone wires in $\Gamma_a$, each individually equivalent to a straight wire in $\Gamma_b$, at least one pair of tb- (lr-) wires intersect at least twice or one pair of a tb-wire and a lr-wire intersect at least three times.

Assume a pair of adjacent $y$-monotone tb-wires intersect at least twice (see Fig.~\ref{fig:proper}(a)).
Consider the top-most two intersections.
The region enclosed by the wires cannot contain vertices as both wires partition the vertices equivalently to $\Gamma_b$.
As the enclosed region is simple and every edge is intersected at most once by a single wire, the order in which edges are intersected along the different wires is the same.
We can locally reroute both wires along the enclosed region to remove both intersections. Repeated application removes all intersections. Contradiction.

Assume an $x$-monotone lr-wire $w_\rightarrow$ and a $y$-monotone tb-wire $w_\downarrow$ intersect at least three times.
Consider a region $R$ enclosed between two consecutive intersections $p_1,p_2$.
Assume w.l.o.g. that $w_\rightarrow$ intersects $w_\downarrow$ left to right in $p_1$ (see Fig.~\ref{fig:proper}(b)).
If $R$ does not contain vertices, then consider the left-most, lowest $x$-monotone increasing path $\pi$ through $R$.
As the boundary right of $\pi$ is $y$-monotone decreasing, $\pi$ must also be $y$-monotone decreasing.
Reroute both wires along $\pi$ between $p_1$ and $p_2$ to remove intersection $p_1$ and $p_2$.

Any region enclosed between two remaining intersections must contain at least one vertex.
Consider the leftmost three consecutive intersections $p_1,p_2,p_3$ along $w_\rightarrow$ .
We have $p_1.x \le p_2.x \le p_3.x$.
Assume for contradiction $w_\downarrow$ crosses through $p_2$ first (last) out of these three intersections.
Then $w_\rightarrow$ upto $p_2$, $w_\downarrow$ upto $p_2$, and the boundingbox around the drawing enclose a simple region.
Wire $w_\downarrow$ enters this region through $p_1$, but cannot exit it without intersecting $w_\downarrow$ somewhere before $p_2$.
Contradiction, as $p_1,p_2,p_3$ are the first three intersection along $w_\rightarrow$.
Thus either $p_1.y \le p_2.y \le p_3.y$ or $p_1.y \ge p_2.y \ge p_3.y$.

Assume w.l.o.g. $p_1.y \ge p_2.y \ge p_3.y$ (see Fig.~\ref{fig:proper}(c)).
The wires between these intersections enclose two disjoint regions $R_1,R_2$.
Each region contains at least one vertex, let $u\in R_1$ and $v\in R_2$. 
Subdivide the plane into four axis-aligned quadrants at $p_2$.
Region $R_1$ lies in the top-left quadrant and $R_2$ in the bottom-right quadrant.
Thus, $u.x<v.x$ and $u.y>v.y$ in $\Gamma_a$.
As the wires are equivalent to $\Gamma_b$, by construction $u.x>v.x$ and $u.y<v.y$ in $\Gamma_b$.
However, by Lemma~\ref{lem:xORy} vertices $u,v$ cannot be both inverted along both axes.
Contradiction.

We conclude that an equivalent set of $x$- respectively $y$-monotone wires exists in $\Gamma_a$ matching the straight-line wiregrid in $\Gamma_b$.
By definition $\Gamma_a$ has spirality one relative to $\Gamma_b$ (with respect to the described set of wires).
\end{proof}

\bigskip\noindent {\bf Lemma~\ref{lem:singleMorph}.\ }
\emph{Let $\Gamma_I$ and $\Gamma_O$ be two degenerate shape-equivalent drawings, where $\Gamma_I$ has spirality one.
There exists a single linear morph from $\Gamma_I$ to $\Gamma_O$ that maintains planarity and orthogonality.}

\begin{proof}
We consider the following linear morph between $\Gamma_I$ and $\Gamma_O$ and prove it has the desired properties.
The partition of the drawing by all wires defines \emph{cells}; regions of the plane not split by any wire.
For each cell containing at least one bend or vertex, linearly interpolate all vertices and bends in $\Gamma_I$ to the unique vertex or bend location in $\Gamma_O$.
This directly defines a linear interpolation for each point (not necessarily a vertex or bend) between $\Gamma_I$ and $\Gamma_O$.

First, we prove that during the described linear morph the drawing remains orthogonal.
The endpoints of all (zero-length) segments crossing a tb-wire have the same $y$- coordinates in $\Gamma_I$ and $\Gamma_O$, hence they remain horizontal.
Symmetrically all segments crossing a lr-wire remain vertical.
All other segments morph to a single point and remain horizontal or vertical as well.

Second, we prove that during the described linear morph the drawing remains planar.
Assume for contradiction there exist two distinct points $p,q$ on an edge or vertex of the drawing that coincide during the linear interpolation (excluding $\Gamma_I,\Gamma_O$).
By linear motion the $x$-coordinates and $y$-coordinates of $p$ and $q$ change linearly.
To be identical at a time $0<t<1$ during the morph $p$ and $q$ must be both $x$- and $y$-inverted in $\Gamma_O$ compared to $\Gamma_I$.
Note that this is not excluded as we do not base our analysis on linear slides.

Assume that $p.x < q.x$ and $p.y < q.y$ in $\Gamma_I$.
The case where either the $x$- or $y$-coordinates are identical in $\Gamma_I$ and $\Gamma_O$ works similarly.
Distinguish whether $p$ and $q$ are on a horizontal or vertical segment.
We will work out the first case and indicate the setup for the other cases, which are analogous.

First, assume $p$ and $q$ are both on a vertical segment in $\Gamma_O$ (see Fig.~\ref{fig:monotone}(a)).
Let $r$ be the top endpoint of the segment containing $p$ and $s$ the bottom endpoint of the segment containing $q$.
In $\Gamma_O$ we have $r.y > s.y$ and $r.x>s.x$.
As $r$ and $s$ have distinct $x$- and $y$-coordinates they are split by at least one tb-wire and one lr-wire in $\Gamma_O$.

The matching (monotone) wires in $\Gamma_I$ split $r$ and $s$ identically.
Furthermore, vertical segments $\overline{pr}$ and $\overline{qs}$ also exist in $\Gamma_I$, though they may have zero-length as they were introduced at a wire-intersection.
Thus $p.y \le r.y$ and $s.y \le q.y$ in $\Gamma_I$.
In $\Gamma_I$ vertices $r,s$ are split by a $y$-monotone tb-wire.
Therefore, it must also be that $r.y < s.y$ in $\Gamma_I$.
Similarly, as $r$ and $s$ are split by an $x$-monotone lr-wire, $r.x<s.x$.
However, then the lr-wire and tb-wire must cross at least three times.
Contradiction.

Second, assume $p$ is on a horizontal segment and $q$ on a vertical segment (see Fig.~\ref{fig:monotone}(b)).
Let $r$ be the right endpoint of the segment containing $p$ and $s$ be the bottom endpoint of the segment containing $q$.
Third, assume $p$ and $q$ are both on a horizontal segment (see Fig.~\ref{fig:monotone}(c)).
Let $r$ be the right endpoint of the segment containing $p$ and $s$ be the left-endpoint of the segment containing $q$.

We conclude there do not exist two distinct points $p,q$ on the edges (vertices) of the drawing that coincide during the linear morph.
\end{proof}

\begin{figure}[b]
\centering
\includegraphics[scale=0.85]{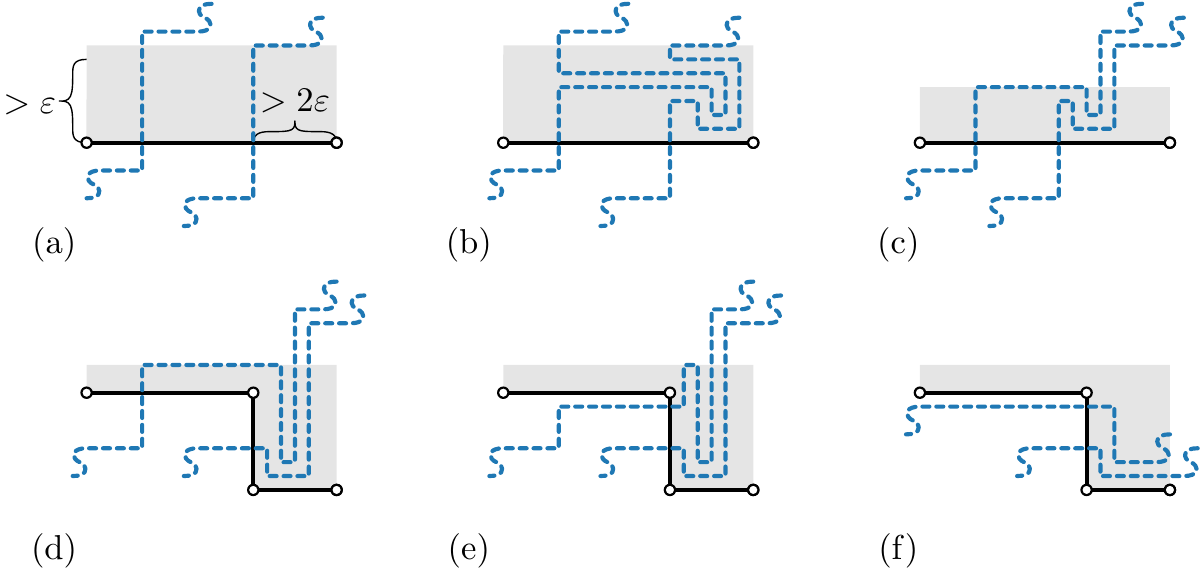}
\caption{\subcap{a} An $\varepsilon$-band adjacent to the edge. \subcap{b} Rerouting all crossing wires inside the $\epsilon$-band while ensuring the wires always pass the last bend-point of the right-most wire. \subcap{c} The wire structure is maintained when the crossing links are reduced. \subcap{d-e} After reducing the right-most link and introducing two bends in $e$ we can safely reroute the remaining wires without increasing spirality as the area above $e$ is completely empty. \subcap{f} The structure is invariant, after reducing all spirality $s$ links the matching structure is left for spirality $s-1$.}
\label{fig:rerouting_new}
\end{figure}

\bigskip\noindent {\bf Lemma~\ref{lem:bandInvariant}.\ }
\emph{At the start of iteration $i$ of the morph, all wires crossing an edge $e$ with links of spirality $i$ form an $i$-windmill in an empty $\varepsilon$-band next to $e$.}

\begin{proof}
At the start of the morph $i=s+1$.
As no edge is crossed by links of spirality $s+1$ the statement holds vacuously.

For the maintenance assume the statement holds at the start of iteration $k\le s+1$.
We prove the statement for $k-1$.
Assume w.l.o.g. that some edge $e$ is crossed by links with positive spirality $k$.
By hypothesis there is a $k$-windmill next to $e$ inside an otherwise empty $\varepsilon$-band (see Fig.~\ref{fig:rerouting_new}(c)).
As a corollary of Lemma~\ref{lem:order}, any slides along links of the wires outside of the $\varepsilon$-band do not affect the $k$-windmill next to $e$.
Thus we only concern ourselves with the local structure.

We reduce the right-most link crossing $e$ and reroute the remaining crossing wires (see Fig.~\ref{fig:rerouting_new}(d)).
This maintains all links in the windmill except for the first link of the right-most wire.
Performing linear slides on the remaining spirality $k$ links within the $\varepsilon$-band also removes the first link in the windmill for the other wires.
Reducing all other spirality $k$ links also removes the last link in the spiral for all wires.
This leaves a $(k-1)$-windmill inside an otherwise empty $\varepsilon$-band next to $e$ (see Fig.~\ref{fig:rerouting_new}(f)).
\end{proof}

\bigskip\noindent {\bf Lemma~\ref{lem:spirOneRerouting}.\ }
\emph{
Let $\Gamma_s$ be the first drawing of an iteration and $\Gamma^r_{s-1}$ the rerouted last drawing.
The spirality of $\Gamma_s$ relative to $\Gamma^r_{s-1}$ is one.}

\begin{proof}
We show that rerouting wires does not destroy staircases.
When rerouting the link crossing $e$ and the first link after the crossing are replaced by four links that together span the same width and the same height.
W.l.o.g. consider the horizontal link and its replacement links.
W.l.o.g. assume the link has positive spirality and that it was part of a downwards staircase (see Fig.~\ref{fig:staircases_the_sequal}(a)).
The two links replacing it span the same width and have the same spirality.
Replacing the original link in the staircase by the two new introduced links maintains the property of the staircase that the left endpoints are $x$-monotone increasing and $y$-monotone decreasing.
Clearly the projection of the two introduced links touches.
Furthermore, as they jointly span the same width they must overlap with the previous and next links in the staircase. (see Fig.~\ref{fig:staircases_the_sequal}(b)).
If the projection of non-consecutive links in the staircase overlap, then a subset of the links exists where all staircase properties are satisfied.

As rerouting links maintains staircases Lemma~\ref{lem:staircase} and therefore Lemmata~\ref{lem:xORy}-\ref{lem:spirOne} still hold and the spirality of $\Gamma_{s}$ with respect to $\Gamma_{s-1}$ is one.
\end{proof}

\begin{figure}[h]
\centering
\includegraphics[scale=0.85]{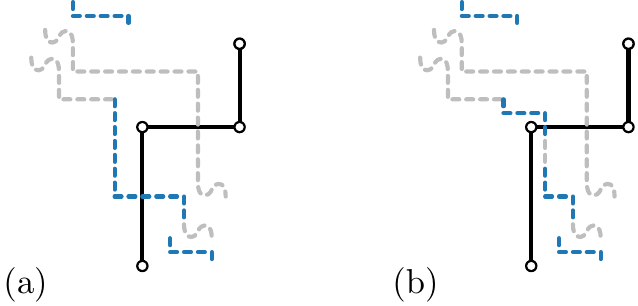}
\caption{\subcap{a} An existing staircase before rerouting (blue).\subcap{b} The rerouted wires may extend the staircase but never destroy it.}
\label{fig:staircases_the_sequal}
\end{figure}

\newpage

The following technical lemma is based on Lemma~12 in~\cite{vanGoethemVerbeek2018}.

\begin{lemma}\label{lem:relativeAngle}
Let $v$ be a vertex with at least two outgoing edges $e,f$ and let $c$  be the turn made at $v$ going from $f$ to $e$, where $c=-1$ for a right turn, $c=1$ for a left turn, and $c=0$ otherwise.
Let $\ell_e$ be a link crossing $e$ and $\ell_f$ be a link crossing $f$.
We have $s(\ell_e)-c+1 \ge s(\ell_f) \ge s(\ell_e)-c-1$.
\end{lemma}

\begin{proof}
Edge $e$ and $f$ either are both horizontal (vertical) in $\Gamma_O$ or not.

For the case where the edges have different orientations, w.l.o.g. assume $e$ is a left-outgoing edge for $v$ and $f$ a top-outgoing edge for $v$ in $\Gamma_O$.
By construction $e$ and $f$ are intersected by a pair of wires $w\in W_\rightarrow$ and $w'\in W_\downarrow$, and they cross before crossing $e$ respectively $f$.
Wires $w$ and $w'$ together with edges $e$ and $f$ must then enclose a simple region in $\Gamma_O$.
As the wires in $\Gamma_I$ form an equivalent set this simple region also exists $\Gamma_I$, however, the shape may be different, moreover, the orientation of the outgoing edges at $v$ may be different.

This cycle by construction contains at least three left turns.
Two at the crossing of the wires with $e$ and $f$, and one at the crossing of the wires.
The turn at $v$ depends on the configuration of $e$ and $f$ in $\Gamma_I$.
Let $\ell_e,\ell_f$ be the links of $w,w'$ crossing $e$ and $f$.
Furthermore, let $k$ be the spirality of the links of $w$ and $w'$ at the crossing between $w$ and $w'$.
As the number of left turns is four larger than the number of right turns when traversing a cycle counter-clockwise we have $(k-s(\ell_e))+(s(\ell_f)-k)+3+c=4$, simplified $s(\ell_f)=s(\ell_e)-c+1$.
When, in  $\Gamma_O$, $e$ is clockwise adjacent to $f$ at $v$ then we get $s(\ell_f)=s(\ell_e)-c-1$.

For the case where both edges are horizontal (vertical) a similar argument holds, but now the cycle is formed by two wires from $W_\rightarrow$ and one wire from $W_\downarrow$ resulting in one more left turn.
We obtain $s(\ell_f) = s(\ell_e)-c$.

Combining all bounds gives the result.
\end{proof}

\bigskip\noindent {\bf Lemma~\ref{lem:SafeRedraw}.\ }
\emph{
We can redraw all edges in $\Gamma_{s-1}$ that were crossed by a maximum-spirality link in $\Gamma_s$ within $6\varepsilon$-boxes while maintaining planarity of the drawing.}

\begin{proof}
Trivially planarity violations occur only inside the $6\varepsilon$-boxes.
Let $\Gamma^r_{s-1}$ be the redrawn version of $\Gamma_{s-1}$.
There are two possible cases causing a planarity violation in $\Gamma^r_{s-1}$.
First, two perpendicular edges leaving $v$ coincide internally after the redraw step.
This occurs if one of the edges is crossed by links of absolute spirality $s$ and the other is not.
Second, two edges leaving $v$ in opposing direction coincide internally after the redraw step.
This occurs if one of the edges is crossed by links of spirality $s$ and the other by links of spirality $-s$.

For the first case, w.l.o.g. assume $e$ is an right-outgoing edge of $v$ and $f$ a bottom-outgoing edge.
Let $\ell_e$ be a link crossing $e$ and $\ell_f$ a link crossing $f$ in $\Gamma_s$.
Assume w.l.o.g. $s(\ell_e)=s$ and $s(\ell_f) < s$.
By Lemma~\ref{lem:relativeAngle}, using $c=-1$, in $\Gamma_s$ we have $s(\ell_e)+2 \ge s(\ell_f)\ge s(\ell_e)$.
Specifically, as $s(\ell_e)=s$ and no larger spiralities exists in $\Gamma_s$, we must have $s(\ell_f)=s(\ell_e)$. Contradiction.

For the second case, w.l.o.g. assume $e$ is a right-outgoing edge of $v$ and $f$ a left-outgoing edge.
Let $\ell_e$ be a link crossing $e$ and $\ell_f$ a link crossing $f$, where $s(\ell_e) = s$ and $s(\ell_f)=-s$.
Using Lemma~\ref{lem:relativeAngle}, with $c=0$, we get $s(\ell_f) \ge s(\ell_e)-1 \ge 0 > -s$. Contradiction.
\end{proof}

\bigskip\noindent {\bf Lemma~\ref{lem:straightLineDrawing}.\ }
\emph{
If $\Gamma_s$ is a straight-line drawing with spirality $s>0$ then there exists a straight-line drawing $\Gamma'_{s-1}$ with spirality $s-1$.}

\begin{proof}
There exists a morph from $\Gamma_s$ to $\Gamma_O$ monotonously decreasing spirality while using rerouting to prevent excess complexity in the edges.
Let $\Gamma^r_{s-1}$ be the first drawing with spirality $s-1$ in this morph.
Consider an edge $e$ that is crossed by maximum absolute spirality links with positive spirality in $\Gamma_s$.
Mirror right and left turns for edges crossed by maximum absolute spirality links with negative spirality.
Edge $e$ has three segments in $\Gamma^r_{s-1}$ that are joined by a right turn followed by a left turn.

Redraw $e$ within the $6\varepsilon$-boxes near the endpoints to create two left turns and a right turn at the start of the edge, and one left turn and two right turns at the end of the edge (see Fig.~\ref{fig:rerouteEdgesLocally}(b)).
Thus the bends in $e$ in $\Gamma^r_{s-1}$ can be encoded as $LLR~RL~LRR$, where $L$ encodes a left turn and $R$ a right turn.
Split differently we have $LLRR~LLRR$ where any wire crossing the edge crosses in the segment between the two groups of bends.

We can remove a pair of consecutive bends $LR$ by performing a zigzag-removing slide on the segment between the bends.
As any such segment is not crossed by wires this does not introduce new bends in the wires.
The result is a straight-line version $\Gamma'_{s-1}$ of $\Gamma^r_{s-1}$.
As no new bends were introduced in any of the wires, the spirality is still $s-1$.
\end{proof}

\bigskip\noindent {\bf Lemma~\ref{lem:straightSpirOne}.\ }
\emph{The spirality of $\Gamma'_{s-1}$ relative to $\Gamma_{s}$ is one.}

\begin{proof}
Let the \emph{main wire set} be the set of wires used to compute the morph including rerouting from $\Gamma_s$ to $\Gamma^r_{s-1}$.
Consider a \emph{reference wire grid} that is a straight-line wire grid in $\Gamma_s$.
We use the reference wire grid to prove the spirality of $\Gamma'_{s-1}$ relative to $\Gamma_s$ is one.
Using Lemma~\ref{lem:equivalent},~\ref{lem:spirOne},~\ref{lem:spirOneRerouting} but swapping the roles of $\Gamma_a$ and $\Gamma_b$, we obtain the result that there is an equivalent monotone set of wires in $\Gamma^r_{s-1}$ matching the reference grid in $\Gamma_s$.
(We remark that the argument for Lemma~\ref{lem:spirOne} is slightly easier in this direction as no coincident vertices exist in $\Gamma_a$.)
Let a segment in $\Gamma^r_{s-1}$ that does not cross a wire of the main wire set be a \emph{free} segment.
Free segments have the same orientation in $\Gamma^r_{s-1}$ as in $\Gamma_s$.

We consider straightening $\Gamma^r_{s-1}$ while deforming the set of reference wires along.
The resulting deformed set of wires for $\Gamma'_{s-1}$ is equivalent to the straight-line reference grid in $\Gamma_s$. What is left to prove is that these wires are still $x$- ($y$-)monotone and hence the spirality of $\Gamma'_{s-1}$ relative to $\Gamma_s$ is one.

If a horizontal (vertical) free segment is crossed by a \emph{reference} wire in $\Gamma^r_{s-1}$, then it must be crossed by a tb- (lr-)wire as the orientation of the segment is the same in $\Gamma_s$ and the reference grid consists of straight lines in $\Gamma_s$.
Straightening the segment using a zigzag-removing slide introduces two bends and a new horizontal link in the tb-wire.
Introducing a horizontal link cannot break $y$-monotonicity of the wire though.
Thus deforming the drawing and the reference grid while simplifying maintains the spirality of the reference grid.
But then there exists an equivalent set of monotone wires in $\Gamma'_{s-1}$ matching the straight-line wires in $\Gamma_s$.
We conclude that the spirality of $\Gamma'_{s-1}$ is one relative to $\Gamma_s$.
\end{proof}
\end{document}